%% file: GSI.tex
\pdfoutput=1
\documentclass[conference]{IEEEtran}
\usepackage{cite}
\usepackage{amsmath,amssymb,amsfonts}
\usepackage[noend]{algorithmic}
\usepackage{graphicx}
\usepackage{textcomp}
\usepackage{xcolor}
\def\BibTeX{{\rm B\kern-.05em{\sc i\kern-.025em b}\kern-.08em
    T\kern-.1667em\lower.7ex\hbox{E}\kern-.125emX}}

\usepackage{balance}  
\usepackage{ntheorem}
\usepackage{times}
\usepackage{epsfig,subfigure}
\usepackage[linesnumbered,ruled,noend]{algorithm2e}
\usepackage{graphicx,floatrow}
\usepackage{listings}
\usepackage{tikz}
\usepackage[T1]{fontenc}
\usepackage{pgfplots}
\usepackage[colorlinks,linkcolor=black,anchorcolor=black,citecolor=black]{hyperref}
\usepackage{booktabs}
\usepackage{multirow,multicol}
\usepackage{threeparttable}
\usepackage{comment}

\newcommand{\Paragraph}[1]{~\vspace*{-0.9\baselineskip}\\{\bf #1}}
\newcommand{\nop}[1]{}
\newtheorem{definition}{Definition}

\newtheorem{example}{Example}
\newtheorem{claim}{Claim}

\newcommand{\optionshow}[2]{#2}
\newcommand{\picfolder}{./}

\nop{

}

\usepackage{floatrow}
\newfloatcommand{capbtabbox}{table}[][\FBwidth]

\usepackage{soul}
\soulregister\cite7 
\soulregister\citep7 
\soulregister\citet7
\soulregister\ref7 
\soulregister\pageref7 
\newcommand{\myhl}[1]{#1}

\begin{document}


\title{GSI: GPU-friendly Subgraph Isomorphism}

\nop{
\author{\IEEEauthorblockN{Li Zeng}
\IEEEauthorblockA{\textit{Peking University} \\
\textit{Peking University}\\
Beijing, China \\
li.zeng@pku.edu.cn }
\and
\IEEEauthorblockN{Lei Zou}
\IEEEauthorblockA{\textit{Peking University} \\
\textit{Peking University}\\
Beijing, China \\
zoulei@pku.edu.cn }
\and
\IEEEauthorblockN{M. Tamer {\"O}zsu}
\IEEEauthorblockA{\textit{University of Waterloo} \\
\textit{University of Waterloo}\\
Canada \\
tamer.ozsu@uwaterloo.ca }
\and
\IEEEauthorblockN{Lin Hu}
\IEEEauthorblockA{\textit{Peking University} \\
\textit{Peking University}\\
Beijing, China \\
hulin@pku.edu.cn }
\and
\IEEEauthorblockN{Fan Zhang}
\IEEEauthorblockA{\textit{Peking University} \\
\textit{Peking University}\\
Beijing, China \\
zhangfanau@pku.edu.cn }
}
}

\author{%
	{ {Li Zeng${^\dag}$}, Lei Zou{${^\dag}$},  M. Tamer {\"O}zsu{${^\ddag}$},  Lin Hu{${^{\dag}}$}, Fan Zhang{${^{\dag}}$}}%
	\\
	\fontsize{10}{10}\selectfont\itshape $~^{\dag}$Peking University, China;
	\fontsize{10}{10}\selectfont\itshape $~^{\ddag}$University of Waterloo,
	Canada;
	\\
	\fontsize{10}{10}\selectfont\ttfamily\upshape $~^{\dag}$$\{$li.zeng,zoulei,hulin,zhangfanau$\}$@pku.edu.cn, $~^{\ddag}$tamer.ozsu@uwaterloo.ca
	\\}


\maketitle

\begin{abstract}
Subgraph isomorphism is a well-known \emph{NP-hard} problem that is widely used in many applications, such as social network analysis and querying over the knowledge graph. 
Due to the inherent hardness, its performance is often a bottleneck in various real-world applications. 
We address this by designing an efficient subgraph isomorphism algorithm leveraging features of GPU architecture, such as massive parallelism and memory hierarchy.
Existing GPU-based solutions adopt two-step output scheme, performing the same join twice in order to write intermediate results concurrently. 
They also lack GPU architecture-aware optimizations that allow scaling to large graphs. 
In this paper, we propose a \emph{G}PU-friendly \emph{s}ubgraph \emph{i}somorphism algorithm, \emph{GSI}. 
Different from existing edge join-based GPU solutions, we propose a \emph{Prealloc-Combine} strategy based on the vertex-oriented framework, which avoids joining-twice in existing solutions. 
Also, a GPU-friendly data structure (called  \emph{PCSR}) is proposed to represent an edge-labeled graph. 
Extensive experiments on both synthetic and real graphs show that GSI outperforms the state-of-the-art algorithms by up to several orders of magnitude and has good scalability with graph size scaling to hundreds of millions of edges.

\nop{
to make breakthrough in three major aspects: \emph{Reducing Work Complexity}, \emph{Memory Latency} and \emph{Load Imbalance}.
The main optimizations include \emph{PCSR} structure for fast $N(v,l)$ extraction and \emph{Prealloc-Combine} strategy, a better substitution of traditional \emph{two-step output scheme}.
We also propose other optimizations such as signature-based filter, batch implementation of set operations, write cache and so on.
Extensive experiments on both synthetic and real graphs show that \emph{GSI} outperforms the state-of-the-art algorithms by up to several orders of magnitude and has a good scalability with varying sizes of data graphs. }
\end{abstract}





\input{introduction}
\input{preliminary}

\input{overview}

\input{pcsr}

\input{join}

\input{optimization}

\input{extension}

\input{experiment}

\input{related}
\input{conclusion}





\bibliographystyle{IEEEtran}
\bibliography{gsi}  




\end{document}

%% file: introduction.tex
\section{Introduction}\label{sec:introduction}

Graphs have become increasingly important in modeling complicated structures and schema-less data such as chemical compounds, social networks and RDF datasets. 
The growing popularity of graphs has generated many interesting data management problems. 
Among these, subgraph search is a fundamental problem: how to efficiently enumerate all subgraph isomorphism-based matches of a query graph over a data graph.
This is the focus of this work. 
Subgraph search has many applications, e.g., chemical compound search \cite{yan2004graph} and search over a knowledge graph \cite{perez2009semantics,lassila1998resource, DBLP:journals/fcsc/ZengZ18}.
A running example (query graph $Q$ and data graph $G$) is given in Figure \ref{fig:example} and Figure \ref{fig:example}(c) illustrates the matches of $Q$ over $G$. 

Subgraph isomorphism is a well-known \textit{NP-hard} problem \cite{DBLP:books/fm/GareyJ79} and most solutions follow some form of tree search with backtracking \cite{DBLP:journals/ijprai/ConteFSV04}.
\nop{
Most subgraph isomorphism algorithms follow some form of tree search with backtracking strategy \cite{DBLP:journals/ijprai/ConteFSV04}, a depth-first search method that finds solutions in a memory efficient manner. Generally, during each step a partial match (initially empty) is iteratively expanded by adding to it new pairs of matched nodes; the pair is chosen using some necessary conditions that ensure the subgraph isomorphism constraints with respect to the nodes mapped so far, and usually also using some heuristics to prune as early as possible unfruitful search paths. Eventually, either the algorithm finds a complete match, or it reaches a state where the current partial mapping cannot be further expanded because of the matching constraints. 
Then, the algorithm will backtrack to probe other search paths.
}
Figure \ref{fig:tree} illustrates the search space for $Q$ over  $G$ of Figure \ref{fig:example}. 
Although existing algorithms propose many pruning techniques to filter out unpromising search paths \cite{DBLP:journals/pvldb/KimSHHC15,CFL-Match}, due to the inherent NP-hardness, the search space is still exponential. 
Therefore, scaling to large graphs with millions of nodes is challenging.
One way to address this challenge is to employ hardware assist.

In this paper, we propose an efficient GPU-based subgraph isomorphism algorithm to speed up subgraph search by leveraging massively parallel processing capability of GPU to explore the search space in parallel. 
Note that our proposed accelerative solution is orthogonal to pruning techniques in existing algorithms \cite{DBLP:journals/pvldb/QiaoZC17, VF3, DBLP:conf/sigmod/HanLL13, DBLP:journals/pvldb/KimSHHC15, DBLP:journals/pvldb/RenW15, CFL-Match, DBLP:journals/pvldb/ShangZLY08, DBLP:journals/pvldb/ZhaoH10}. 

\vspace{-0.1in}
\begin{figure}[htbp]
	\centering
	\includegraphics[width=8cm]  {\picfolder 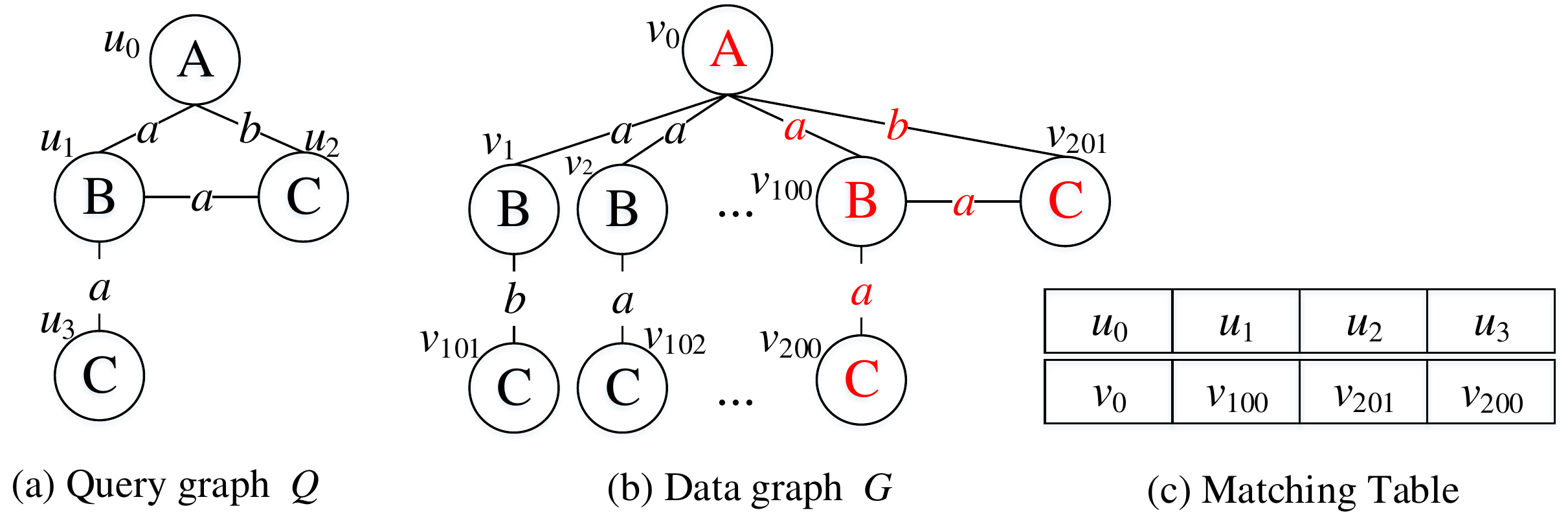}        	
	\caption{An example of Query Graph and Data Graph}      	
	\label{fig:example}  
    \vspace{-0.1in}     	
\end{figure}
\vspace{-0.2in}

\begin{figure}[htbp]
	\centering
	\includegraphics[width=8cm]  {\picfolder 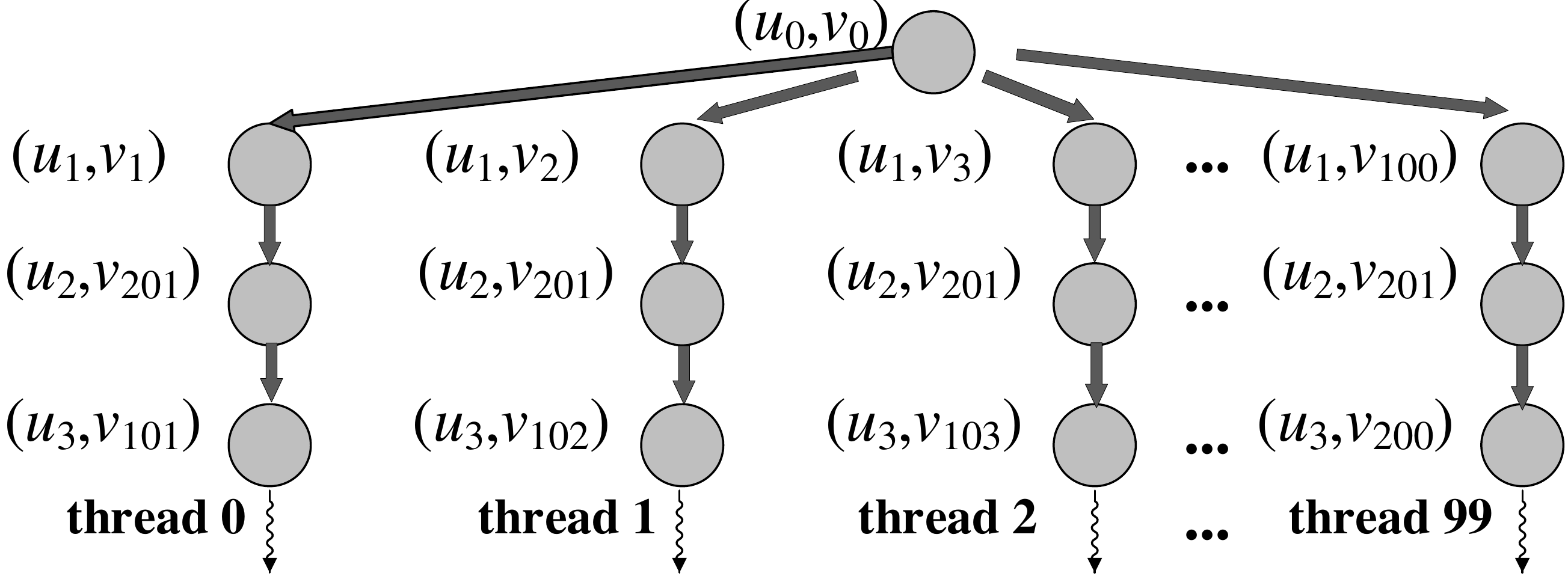}        	
    \caption{An example of searching tree of $Q$ in $G$}
	\label{fig:tree}     
    \vspace{-0.1in}     	
\end{figure}
\vspace{0.05in}

\begin{figure*}[htbp]
	\centering
	\includegraphics[width=16cm]  {\picfolder 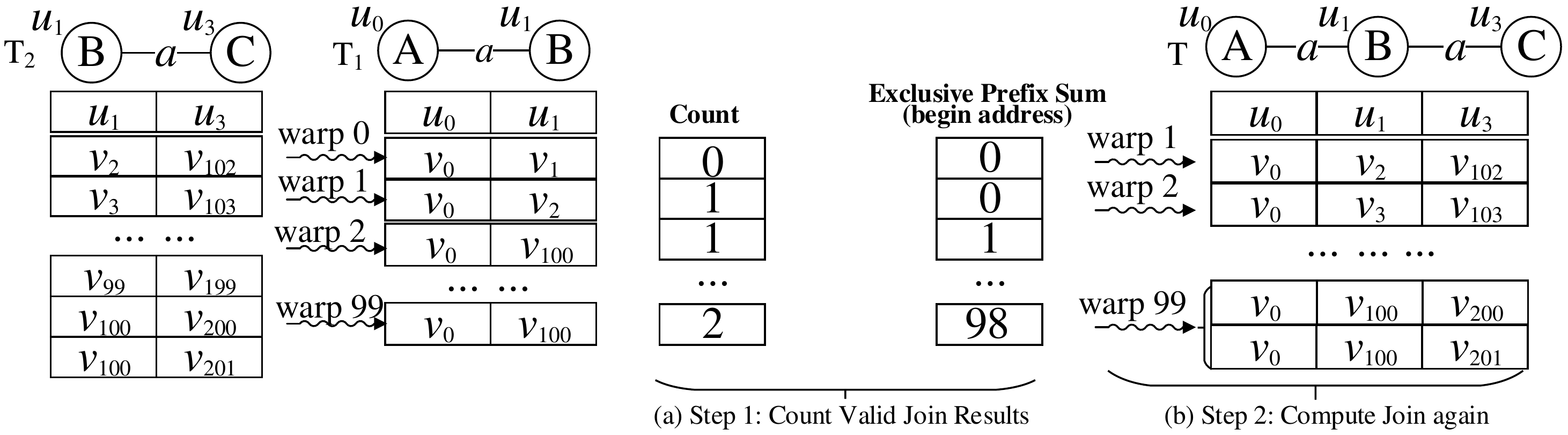}        	
    \caption{An example of ``two-step output scheme''}
	\label{fig:2step}     
    \vspace{-0.1in}     	
\end{figure*}
\vspace{-0.2in}

\nop{
Although GPU has the power of massive parallelism, badly designed GPU algorithms may be even  slower than CPU counterparts due to thread underutilization, higher memory latency and irregular access patterns. 
As mentioned above, the backtracking paradigm finds efficient solutions for CPU-based subgraph isomorphism algorithms, but it cannot be efficiently adapted to GPUs due to the \emph{warp divergence} and \emph{uncoalesced memory access} \cite{DBLP:conf/europar/JenkinsAOCS11,GpSM}. 
GPU operations are based on \emph{warps}, which are groups of threads to be executed in single-instruction-multiple-data (SIMD) fashion. 
Assume that each thread handles one path of the search tree in Figure \ref{fig:tree}. 
Even if some branches end earlier, they still have to wait for other threads in the same warp to finish. 
That is called warp divergence, which is the straggler problem in parallel computing.
Furthermore, if the threads within one warp access scattered memory addresses, called uncoalesced memory access (an example is shown in Figure \ref{fig:uncoalesce}), this leads to more memory transactions and large memory latency. 
In the backtracking paradigm, each thread in one warp always needs to access neighbors of different vertices. 
These are definitely uncoalesced memory accesses. 
}



To the best of our knowledge, two state-of-the-art GPU-based subgraph isomorphism algorithms exist in the literature: GpSM \cite{GpSM} and GunrockSM \cite{GunrockSM}. 
In order to avoid the bottlenecks of backtracking \cite{DBLP:conf/europar/JenkinsAOCS11}, they both adopt the breadth-first exploration strategy. 
They perform edge-oriented computation, where they collect candidates for each edge of $Q$ and join them to find all matches.
The edge-based join strategy suffers from high volume of work when implemented on GPU. 
A key issue is how to write join results to GPU memory in a massively parallel manner. 
GpSM and GunrockSM  employ the ``two-step output scheme'' \cite{DBLP:conf/IEEEpact/HeFLGW08}, as illustrated in Example \ref{example:1}.
\nop{
In the first round, the number of join results  is counted for each thread, based on which, memory can be allocated for join result addresses of each thread. 
Then, in the second round, each thread conducts the same join again and writes the results to the corresponding addresses. 
}
\begin{example}\label{example:1} 
Consider $Q$ and $G$ in Figure \ref{fig:example}. 
Tables $T_1$ and $T_2$ in Figure \ref{fig:2step} show the matching edges of $\overline{u_0u_1}$ and  $\overline{u_1u_3}$, respectively. 
In order to obtain matches of the subgraph induced by vertices $u_0$, $u_1$ and $u_3$, GpSM performs the edge join $T_1 \Join T_2$.  
Assume that each processor handles one row in $T_1$ for joining. 
Writing the join results to memory in parallel may lead to a conflict, since different processors may write to the same address. 

To avoid this, the naive solution is locking, but that reduces the parallelism. 
GpSM and GunrockSM  use ``two-step output scheme'' instead. 
In the first step, each processor joins one row in $T_1$ with the entire table $T_2$ and counts valid matches (Figure \ref{fig:2step}(a)). 
Then, based on the prefix-sum, the output addresses for each processor are calculated.
In the second step, each processor performs the same join again and writes the join results to the calculated memory address in parallel (Figure \ref{fig:2step}(b)). 
\end{example}
The two-step output scheme performs the same join twice, doubling the amount of work,  and thus suffers performance issues when GPU is short of threads on large graphs.
In order to avoid joining twice, we propose a \emph{Prealloc-Combine} approach, which is based on \emph{joining candidate vertices } instead of edges.
During each iteration, we always join the intermediate results with a  candidate vertex set. 
To write the join results to memory in parallel, we pre-allocate enough memory space for each row of $M$  and perform the vertex join only once. 
We use vertex rather than edge as the basic join unit, because we cannot estimate memory space for edge join results, which is easy for vertex join. 
More details are given in Section \ref{sec:basicjoin}. 

Vertex join has two important primitive operations: accessing one vertex's neighbors and set operations. 
To gain high performance, we propose an efficient data structure (called \emph{PCSR}, in Section \ref{sec:graphds}) to retrieve a vertex's neighbors, especially for an edge-labeled graph. 
Also, adapting to GPU architecture, we design an efficient GPU-based algorithm for set operations. 

Putting all these together, we obtain an efficient \emph{G}PU-friendly \emph{s}ubgraph \emph{i}somorphism solution (called \emph{GSI}). 
Our primary contributions are the following:
\begin{itemize}
    \item We propose an efficient data structure (PCSR) to represent edge-labeled graphs, which helps reduce memory latency.
    \item Using \emph{vertex-oriented} join, we propose Prealloc-Combine strategy instead of two-step output scheme, which is significantly more performant.
    \item Leveraging GPU features, we discuss efficient implementation of set operations, as well as optimizations including load balance and duplicate removal.
    \item Experiments on both synthetic and real large graph datasets show that GSI outperforms the state-of-the-art approaches (both CPU-based and GPU-based) by several orders of magnitude. 
        Also, GSI has good scalability with graph size scaling to hundreds of millions of edges.
    
    \nop{
    \item \emph{Reducing Work Complexity}. Instead of two step output scheme, the \emph{Prealloc-Combine} strategy reduces work complexity. In addition, a duplication removal  heuristic  is proposed to reduce redundant computing.
    \item \emph{Shortening Memory Latency}. An efficient \emph{PCSR} structure is designed to speed up $N(v,l)$ extraction (extracting the neighbor set of vertex $v$ with adjacent edge label $l$). Shared memory is fully utilized to reduce latency of global memory access, e.g., batch implementation of set operations and write cache.
    \item \emph{Addressing Load Imbalance}. The \emph{4-layer balance scheme} is used to balance the workload in a finer grain.
    \item \emph{Conducting Extensive Experiments.} }

\end{itemize}

\optionshow{}{
The rest of the paper is organized as follows. Section \ref{sec:preliminary} gives formal definitions of subgraph isomorphism and background knowledge.
In Section \ref{sec:algorithm}, we introduce the framework of \emph{GSI}, which consists of filtering phase and joining phase. 
The optimizations of \emph{GSI} are discussed in Section \ref{sec:optimize}. 
Our method can also support other graph pattern semantics, such as homomophism and edge isomorphism;  and can process multi-labeled graphs as well.
Section \ref{sec:extension} gives the details of these extensions.
Section \ref{sec:experiment} shows all experiment results and Related works are presented in Section \ref{sec:related}.
Finally, Section \ref{sec:conclusion} concludes the paper.
}

%% file: preliminary.tex
\section{Preliminaries}\label{sec:preliminary}

\optionshow{}{
\begin{table}[H]
	\small
	\centering
	\caption{Notations}
	\begin{tabular}{|c||p{6cm}|}
		\hline
		$G,Q$& Data graph and query graph, respectively \\
		\hline
		$v,u$& Vertex in $G$ and $Q$, respectively\\
		\hline
		$S(v),S(u)$& Encoding of vertex $v$ or $u$\\
		\hline
		$N(v),N(u)$& All neighbors of vertex $v$ or $u$ \\
		\hline
		$N(v,l)$& Neighbors of vertex $v$ with edge label $l$ \\
        \hline
        $freq(l)$& Frequency of label $l$ in $G$ \\
		\hline
		$C(u)$& The candidate set of query vertex $u$ in $G$ \\
		\hline
		$M$,$M^{\prime}$& The old and new intermediate result table, each row represents a partial answer, each column correspondings to a query variable\\
		\hline
        $num(L)$& The number of (valid) elements in $L$ \\
        \hline
        $D=P(G,l)$& Edge label $l$-partitioned subgraph of $G$ \\
        \hline
	\end{tabular}
	\label{table:notations}
    \vspace{-0.15in}     	
\end{table}
}


\optionshow{}{
In this section, we formally define our problem and review the terminology used throughout this paper. 
We also introduce GPU background and discuss the challenges for GPU-based subgraph isomorphism computation. 
Table \ref{table:notations} lists the frequently-used notations in this paper.
}

\subsection{Problem Definition}\label{sec:definition}

\begin{definition}\label{def:graph} \textbf{(Graph)}
	A graph is denoted as $G=\{V,$ $E,$ $L_{V}$ $L_{E} \}$, where $V$ is a set of vertices; $E \subseteq V \times V$ is a set of undirected edges in $G$; $L_{V}$ and $L_{E}$ are two functions that assign labels for each vertex in $V(G)$ and each edge in $E(G)$, respectively. 
\end{definition}

\nop{
\begin{definition}\label{def:subgraph} \textbf{(Subgraph)}
	Given a graph $G=\{V,$ $E,$ $L_{V},$ $L_{E}\}$, a subgraph of $G$ is denoted as $G^{\prime}=\{V^{\prime},$ $E^{\prime},$ $L^{\prime}_{V},$ $L^{\prime}_{E}\}$, where vertex sets $V^{\prime}$ and edge sets $E^{\prime}$ in $G^{\prime}$ are subsets of $V$ and $E$, respectively, denoted as $V^{\prime} \subseteq V$ and $E^{\prime} \subseteq E$. Furthermore, for vertex and edge label functions, $L^{\prime}_{V^{\prime}} \subseteq L_V$ and $L^{\prime}_{E^{\prime}} \subseteq L_E$. 
\end{definition}
}

\begin{definition}\label{def:graphisomorphism} \textbf{(Graph Isomorphism)}
    Given two graphs $H$ and $G$, $H$ is \emph{isomorphic} to $G$ if and only if there exists a bijective function $f$ between the vertex sets of $G$ and $H$ (denoted as $f:V(H)\longrightarrow V(G) )$, such that 
	\begin{itemize}
	\item $\forall u \in V(H), f(u) \in V(G) $ and $L_{V}(u)=L_{V}(f(u))$, where $V(H)$ and $V(G)$ denote all vertices in graphs $H$ and $G$, respectively. 
	\item $\forall  \overline{u_1u_2} \in E(H)$,  $\overline{f(u_1)f(u_2)} \in E(G)$  and $L_{E}(\overline{f(u_1)f(u_2)})$ $=L_{E}(\overline{u_1u_2})$, where $E(H)$ and $E(G)$ denote all edges in graphs $H$ and $G$, respectively. 	
	\item $\forall  \overline{u_1u_2} \in E(G)$,  $\overline{f^{-1}(u_1)f^{-1}(u_2)} \in E(H)$  and $L_{E}(\overline{f^{-1}(u_1)f^{-1}(u_2)})$ $=L_{E}(\overline{u_1u_2})$
	\end{itemize}
\end{definition}

\begin{definition}\label{def:subgraphisomorphism}\textbf{(Subgraph Isomorphism Search)}
    Given query graph $Q$ and data graph $G$, the subgraph isomorphism search problem is to find all \emph{subgraph}s $G^{\prime}$ of $G$ such that  $G^{\prime}$ is \emph{isomorphic} to $Q$. $G^{\prime}$ is called a \emph{match} of $Q$.
\end{definition}

\nop{
Following Definition \ref{def:graph}, \ref{def:subgraph}, \ref{def:graphisomorphism} and \ref{def:subgraphisomorphism}, the problem studied in this paper is defined in Definition \ref{def:problem}.
Figure \ref{fig:example} shows an example of query graph $Q$ and data graph $G$.
The result of finding $Q$ in $G$ is listed in Table \ref{table:result}, i.e. a single matching.
Table \ref{table:notations} shows some frequently-used notations.
}

This paper proposes an efficient GPU-based solution for subgraph isomorphism search. 
Without loss of generality, we assume $Q$ is connected and use $v$, $u$, $N(v)$, $N(v,l)$, $num(L)$, and $|A|$  to denote a data vertex, a query vertex, all neighbors of $v$, $\{v^{\prime}| \overline{vv'}  \in E(G)  \wedge L_E (\overline{vv^{\prime}} ) = l\}$, the number of currently valid elements in set $L$, and the size of set $A$, respectively.

\nop{
\myhl{
Note that our method can also support other graph pattern semantics, such as homomophism and edge isomorphism;  and can process multi-labeled graphs as well. We give the details in the full version of this work \cite{fullVersion} due to the space limit. 
}
}

\nop{
\myhl{
Our work can also be extended to support homomorphism and edge isomorphism, and can process multi-labeled graphs.
The details are in our full paper \cite{fullVersion}.
}
	Without loss of generality, we assume $Q$ is connected; otherwise, we can regard each connected component of $Q$ as a separate query and execute them individually.
	Unless otherwise specified, we use $v$, $u$, $N(v)$, and $N(v,l)$ to denote a data vertex, a query vertex, all neighbors of $v$, and $\{v^{\prime}| \overline{vv'}  \in E(G)  \wedge L_E (\overline{vv^{\prime}} ) = l\}$, respectively.
}
\nop{
Without loss of generality, we make the following assumptions for ease of presentation:
\begin{itemize}
\item Graphs are undirected. However, our method GSI can be easily extended to process directed graphs.
\item Each query graph $Q$ is connected; otherwise, we can regard each connected component of $Q$ as a separate query and perform them one by one.
\end{itemize}
}

\nop{

As is known to all, subgraph isomorphism-based subgraph search (\emph{subgraph search} for short) is a classical NP-hard problem. 
The hardness of subgraph search lies in the enumeration and validation of all possible matching subgraphs in $G$.
Though equipped with multiple cores and SIMD mechanism, modern CPU is still limited in the power of processing large graphs.
In this work, our goal is to speed up \emph{subgraph search} on large graphs by leveraging the massive parallelism ability of GPU.  

\begin{definition}\label{def:problem}\textbf{Problem Statement.} 
    Given a query graph $Q$ and a data graph $G$, the problem is to design an efficient GPU-based solution to find all subgraph \emph{matching}s of query $Q$ over data graph $G$.
\end{definition}
}

\nop{
\begin{figure}[htbp]
	\centering
	\includegraphics[width=8cm]  {images/example.pdf}        	
	\caption{An example of Query Graph and Data Graph}      	
	\label{fig:example}  
	\vspace{-0.2in}     	
\end{figure}
\begin{table}[H]
	\centering
	\caption{Matchings in Example \ref{fig:example}}
	\begin{tabular}{|c|c|c|c|}
		\hline
		$u_{0}$ & $u_{1}$ & $u_{2}$ & $u_{3}$ \\
		\hline 
		$v_{0}$ & $v_{100}$ & $v_{201}$ & $v_{200}$ \\
		\hline
	\end{tabular}
	\label{table:result}
	\vspace{-0.2in}     	
\end{table}
}

\subsection{GPU Architecture} \label{sec:arch}

GPU is a discrete device that contains dozens of streaming multiprocessors (SM) and its own memory hierarchy.
Each SM contains hundreds of cores and CUDA (Compute Unified Device Architecture) programming model provides several thread mapping abstractions, i.e., a thread hierarchy. 

\Paragraph{Thread Hierarchy}.
Each core is mapped to a \emph{thread} and a \emph{warp} contains 32 consecutive threads running in Single Instruction Multiple Data (SIMD) fashion. 
When a warp executes a branch, it has to wait though only a portion of the threads take a particular branch; this is termed as \emph{warp divergence}.
A \emph{block} consists of several consecutive warps and each block resides in one SM.
Each process launched on GPU (called a \emph{kernel function}) occupies a unique \emph{grid}, which includes several equal-sized blocks.

\Paragraph{Memory Hierarchy}.
In Figure \ref{fig:memory}, \emph{global memory} is the slowest and largest layer.
Each SM owns a private programmable high-speed cache, \emph{shared memory}, that is accessible by all threads in one block. 
Although the size of shared memory is quite limited (Taking Titan XP as example, only 48KB per SM), accessing shared memory is nearly as fast as thread-private registers. 
Access to global memory is done through 128B-size transactions and the latency of each transaction is hundreds of times longer than access to shared memory. 
If threads in a warp access the global memory in a consecutive and aligned manner, fewer transactions are needed. 
For example, only 1 transaction is used in coalesced memory access (Figure \ref{fig:coalesce}) as opposed to 3 in uncoalesced memory access (Figure \ref{fig:uncoalesce}). 

\begin{figure}[htbp]
	\centering
	\includegraphics[width=8cm]  {\picfolder 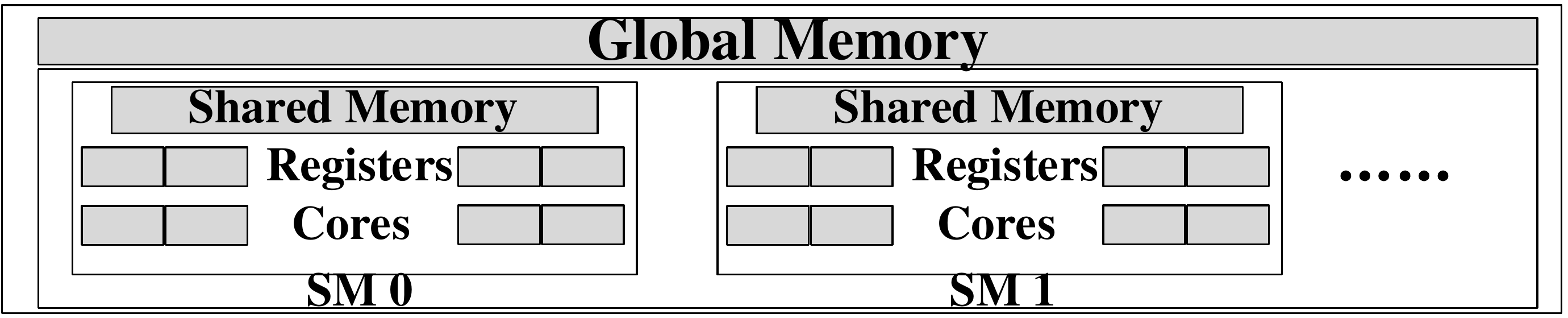}        	
	\caption{Memory Hierarchy of GPU}      	
	\label{fig:memory}   
	\vspace{-0.1in}     	
\end{figure}
\vspace{-0.2in}

\begin{figure}[htbp]
	\centering
	\includegraphics[width=8cm]  {\picfolder 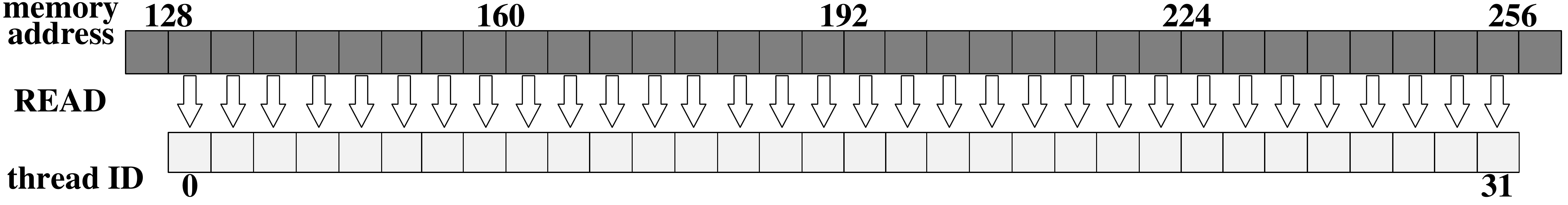}        	
	\caption{An example of coalesced memory access}      	
	\label{fig:coalesce}   
	\vspace{-0.15in}     	
\end{figure}
\vspace{-0.2in}

\begin{figure}[htbp]
	\centering
	\includegraphics[width=8cm]  {\picfolder 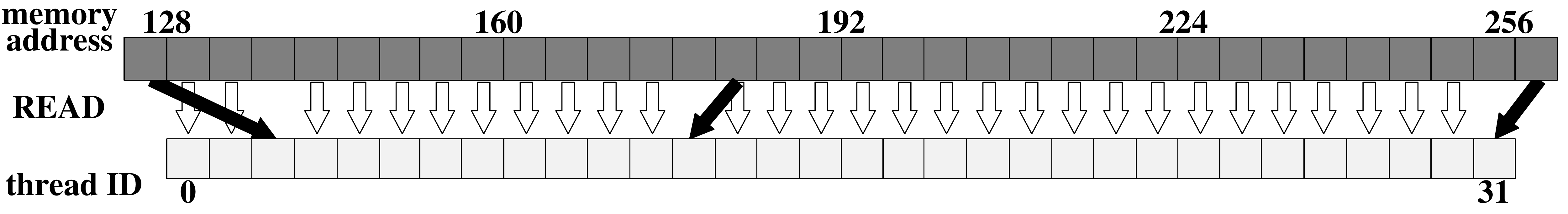}        	
	\caption{An example of uncoalesced memory access}      	
	\label{fig:uncoalesce}     
	\vspace{-0.15in}     	
\end{figure}
\vspace{-0.2in}

\subsection{Challenges of GPU-based Subgraph Isomorphism}\label{sec:challenge}

Although GPU is massively parallel, a naive use of GPU may yield worse performance than highly-tuned CPU algorithms. 
There are three challenges in designing GPU algorithms for subgraph isomorphism that we discuss below.


\nop{
if each thread $t_{j}$ is used for one vertex $v_{j}$, neighbors of $v_{j}$ all needs to be fetched by $t_{j}$ from \emph{Global Memory}.
In this case warp divergence will happen due to the inequality of neighbor number of different vertices.
If thread $t_{i}$ finishes its task first, it must wait for other threads in the same warp to end, as a result of the \emph{SIMD} restriction.
Generally, it is a well-known problem named load imbalance when applying parallel algorithms.
The warp divergence is the result of intra-warp load imbalance, but load imbalance also occurs between warps or even blocks. 
}

\Paragraph{Amount of Work}.
Let $n$ and $n^{\prime}$ be the number of vertices of $G$ and $Q$, the amount of work is \myhl{$n^{n^{\prime}}$} in Figure \ref{fig:tree}.
If there are sufficient number of threads, all paths can be fully parallelized.
But that is not always possible and  too much redundant work will degrade the performance.
GpSM's strategy (filtering candidates and joining them) is better as it prunes invalid matches early.
However, Example \ref{example:1} shows that the two-step output scheme used in GpSM doubles the amount of work in join processing, which is a key issue that must be overcome.

\nop{
There are two metrics of GPU algorithms: work complexity, the number of all steps if executing sequentially; step complexity, the maximum steps needed when running concurrently, which corresponds to time cost.
In the case of parallel prefix-sum \cite{harris2007parallel}, the step complexity is  $O(\log n)$ while the work complexity is $O(n)$.
Considering the case of summing up $n$ integers, sequential algorithm needs $O(n)$ steps and its step complexity and work complexity are both $O(n)$.
In contrary, using parallel operation, the step complexity will be lowered to $O(\log n)$ while the work complexity is still $O(n)$.
However, increasing work complexity may not help reduce the step complexity especially when the amount of work is much larger than the GPU cores.
To utilize parallelism for graph algorithms, sequential algorithms often need to be re-designed and introduce extra work that may increase work complexity but reduce step complexity.
However, if the computing resource is not sufficient for tasks, high work complexity is the key factor that reduces performance.
For example, Merrill \cite{DBLP:conf/ppopp/MerrillGG12} proposes several strategies to reduce the work complexity of BFS on GPU from $O(|V|^{2}+|E|)$ to asymptotically optimal $O(|V|+|E|)$.
In the case of subgraph isomorphism, state-of-the-art techniques adopt two-step output scheme, which doubles the join process, thus dramatically increases the work complexity.
}

\Paragraph{Memory Latency}. 
Large graphs can only be placed in global memory.
In subgraph isomorphism, we need to perform $N(v,l)$ extractions many times, and they are totally scattered due to inherent irregularity of graphs \cite{burtscher2012quantitative}.
It is hard to coalesce memory access in this case, which aggravates latency.

\nop{
The inherent irregularity \cite{burtscher2012quantitative} of graphs makes it hard to coalesce memory access, hence aggravates the latency of access to global memory.
GSI uses an encoding-based strategy for filtering and a novel \emph{PCSR} data structure for extracting neighbors, which help reduce the number of memory transactions.
Furthermore, shared memory is well utilized to cache data, thus accelearte set operations in the join process.
}

\Paragraph{Load Imbalance}. 
GPU performs best when each processor is assigned the same amount of work.
However, neighbor lists vary sharply in size, causing severe imbalance between blocks, warps and threads.
Balanced workload is better, because the overall performance is limited by the longest workload. 




\nop{
\subsection{Comparison and Analysis} \label{sec:existingsolution}

We select two state-of-the-art algorithms as counterparts: GpSM \cite{GpSM} and GunrockSM \cite{GunrockSM}.
They both generate edge candidates first and join them by \emph{two-step output scheme}.
The joining phase is done twice, which doubles the work complexity.
Other than high work complexity, GunrockSM and GpSM both lack optimizations in leveraging  \emph{Shared Memory} to reduce memory latency.
Furthermore, they both have the problem of load imbalance to some degree.

Shortcomings analyzed above are consistant with challenges in Section \ref{sec:challenge} and we redesign GPU-based subgraph isomorphism to consider them.
Traditional two-step output scheme is replaced by a more elaborate method \emph{Prealloc-Combine}) to avoid double work complexity.
The \emph{Duplicate Removal} strategy also helps reduce redundant inputs, thus reduces the amount of work.
To speed up global memory access, \emph{Shared Memory} is fully utilized by \emph{PCSR} data structure for $N(v,l)$ extraction, \emph{Batch Implementation} of set operations and \emph{Write Cache} for writing intermediate results.
Memory accesses are also perfectly coalesced with these techniques.
Besides, we use finer grained load balance strategy(i.e. \emph{4-layer balance scheme}) to limit imbalance in all layers.

To sum up, based on the limitations of two counterparts, we propose our solution \emph{GSI} to accelerate subgraph isomorphism on GPU.
The filtering strategy and joining phase of GSI is very different from GunrockSM and GpSM, which will be detailed in Section \ref{sec:algorithm}.
Without loss of generality, we make some assumptions for ease of presentation:
\begin{itemize}
\item Graphs are undirected.(However, GSI can be naturally extended to process directed graphs.)
\item All queries are weakly connected and no parallel edges exist.
\item The size of work group on GPU is 32 threads, as used in \emph{CUDA 8.0}.
\end{itemize}
}

%% file: overview.tex
\section{Solution Overview}\label{sec:algorithm}

\optionshow{}{
The framework of \emph{GSI} is given in Figure \ref{fig:frame}, which consists of filtering and joining phases. 
}
Our solution consists of filtering and joining phases.
\optionshow{}{
In the filtering phase, a set of candidate vertices in data graph $G$ are collected for each query node $u \in V(Q)$ (denoted as $C(u)$); while, in the joining phase, these candidate sets are joined according to the constraints of subgraph isomorphism (see Definition \ref{def:subgraphisomorphism}). 
We discuss how to use GPU to accelerate both phases in the following sections. 
}

\optionshow{}{
\begin{figure}[htbp]   
	\centering
\includegraphics[width=8cm]{\picfolder 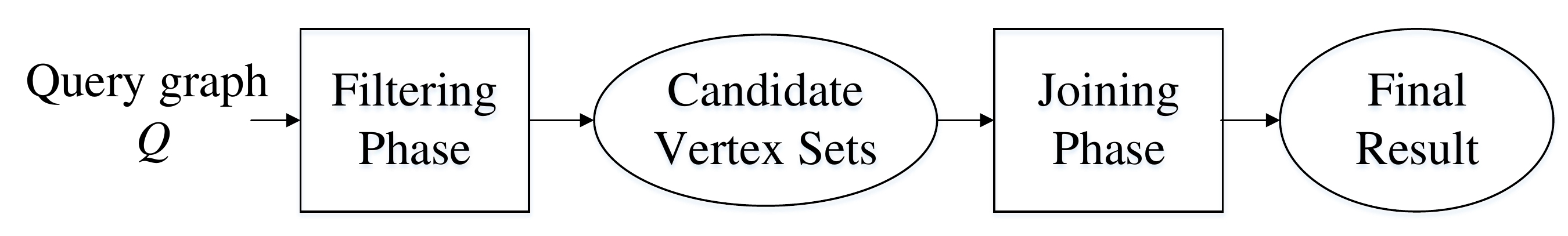}        	
\caption{Framework of GSI algorithm}      	
\label{fig:frame}    
	\vspace{-0.15in}     	
\end{figure}
}

\input{filter}

\begin{figure*}[htbp]       
    \centering          
    \includegraphics[width=16cm]{\picfolder 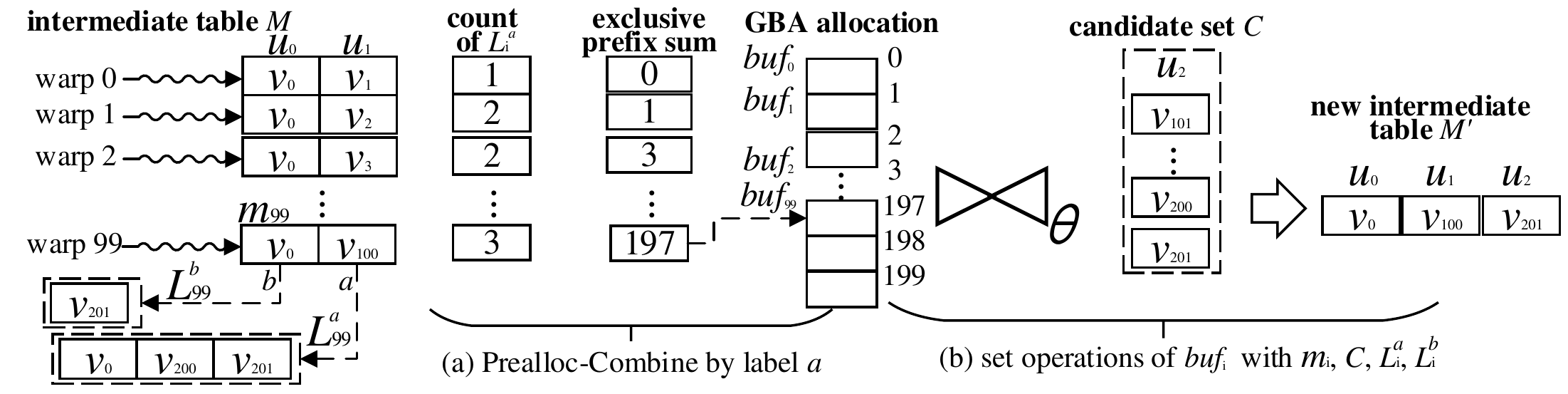}
	\caption{Vertex-oriented Join Strategy}      	
    \label{fig:newtable}
	\vspace{-0.1in}     	
\end{figure*}

\subsection{Joining Phase}\label{sec:join}

The outcome of filtering are candidate sets for all query vertices. 
In Figure \ref{fig:example}, candidate sets are $C(u_{0})=\{v_{0}\}$, $C(u_{1})=\{v_{1},v_{2},...,v_{100}\}$,  and $C(u_{2})=C(u_{3})=\{v_{101},v_{102},...,v_{201}\}$. 
Figure \ref{fig:newtable} demonstrates our vertex-oriented join strategy.
Assume that we have matches of edge $\overline{u_0u_1}$ in table $M$ and candidate vertices $C(u_2)$. 
In $Q$, $u_2$ is linked to $u_0$ and $u_1$ according to the edge labels $b$ and $a$, respectively. 
Thus, for each record $(v_i,v_j)$ in $M$, we read $N(v_i,b)$ and $N(v_j,a)$ and do the set operation $N(v_i ,b) \cap N(v_j ,a) \cap C(u_2 ) \setminus \{v_{i},v_{j}\}$, where $N(v_i,b)$ and $N(v_{j},a)$ denote neighbors of $v_i$ with edge label $b$ and $v_{j}$ with edge label $a$, respectively. 
If the result is not empty, new partial answers can be generated, as shown in Figure \ref{fig:newtable}. 

Notice that there are two primitive operations: accessing one vertex's neighbors based on the edge label (i.e., $N(v,l)$ extraction) and set operations. 
We first present a novel data structure for graph storage on GPU (Section \ref{sec:graphds}). 
Then, the parallel join algorithm (including the implementation of set operations) is detailed in Section \ref{sec:basicjoin}. 

\nop{
In this section, we discuss how to implement these primitives efficiently by employing the power of massive parallelism. 
Firstly, we present a novel data structure in our GPU subgraph join algorithm (Section \ref{sec:graphds}). 
Then, the parallel join algorithm is described in detail in Section \ref{sec:basicjoin}. 
GPU-based set operations will be studied in Section \ref{sec:optimization}.
The implementation of set operations is also critical to the overall performance, which is placed in later sections.

Firstly, we present a novel data structure in our GPU subgraph join algorithm (Section \ref{sec:graphds}). 
Then, the join algorithm is described in detail in Section \ref{sec:basicjoin}. 
Finally, in Section \ref{sec:setop}, we also study the implementation of set operations which is critical to the overall performance.}

%% file: filter.tex
\subsection{Filtering Phase}\label{sec:filter}


Generally, a lightweight filtering method with high pruning power is desirable. 
Since GSI adopts a vertex-oriented strategy, we select candidate vertices $C(u)$ for each query node $u$ in query graph $Q$. 
More powerful pruning means fewer candidates. 
Many pruning techniques have been proposed, such as \cite{DBLP:journals/fcsc/ZengZ18}. 
The basic pruning strategy is based on  ``neighborhood structure-preservation'': if a vertex $v$ in $G$ can match  $u$ in $Q$, the neighborhood structure around $u$ should be preserved in the neighborhood around $v$. 
In this work, we propose a suitable data structure that fits GPU architecture to implement pruning.

We encode the neighborhood structure around a vertex $v$ in $G$ as a length-$N$ bitvector signature $S(v)$. 
Generally, it has two parts. 
The first part is called vertex label encoding that hashes a vertex label into $K$ bits. 
The second part encodes the adjacent edge labels together with the corresponding neighbor vertex. 
We divide the $(N-K)$ bits into $\frac{{N - K}}{2}$ groups with 2 bits per group.
For each (edge, neighbor) pair $(e,v^{\prime})$ of a vertex $v$, we combine $L_{e}$ and $L_{v^{\prime}}$ (i.e., the labels of edge $e$ and $v^{\prime}$) into a key and hash it to some group. Each group has three states: ``00''-- no pair is hashed to this group; ``01''--only a single pair is hashed to this group; and ``11''--more than one pair is hashed to this group. 
Figure \ref{fig:encodeTABLE}(a) illustrates vertex signature $S(v_0)$ of $G$ in Figure \ref{fig:example}. 
We offline compute all vertex signatures in $G$ and record them in a signature table (see Figure \ref{fig:encodeTABLE}(b)). 
We have the same encoding strategy for each vertex $u$ in $Q$. 
It is easy to prove that if $S(v)\&S(u)\neq S(u)$, $v$ is definitely not a candidate for $u$ (``\&'' means ``bitwise AND operation''). 

Given a query graph $Q$, we compute online vertex signatures for $Q$. 
For each query vertex $u$, we have to check all vertex signatures in the table (such as Figure \ref{fig:encodeTABLE}(b)) to fix candidates. 
We can perform the filtering in a massively parallel fashion. 
Furthermore, the natural load balance of accessing fixed-length signatures is suitable for GPU. 
To further improve the performance, we organize the vertex signature table in column-first instead of row-first. 
Recall that all threads in a warp read the first element of different signatures in the table, the row-first layout leads to gaps between memory accesses (see Figure \ref{fig:encodeTABLE}(c)), i.e., these memory accesses cannot be coalesced. 
Instead, the column-first layout provides opportunities to coalesce memory accesses (see Figure \ref{fig:encodeTABLE}(d)).

\begin{figure*}[htbp]
	\centering
	\includegraphics[width=16cm]  {\picfolder 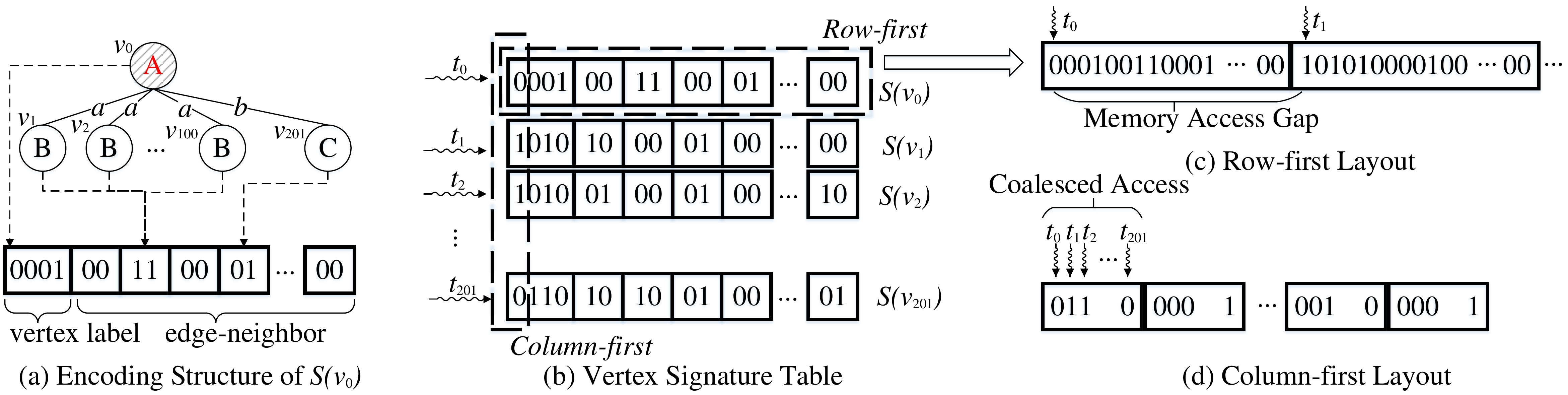}        	
	\caption{Encoding table of data vertices}      	
	\label{fig:encodeTABLE}        
    \vspace{-0.1in}
\end{figure*}

%% file: pcsr.tex
\section{Data Structure of Graph: PCSR}\label{sec:graphds}

\nop{
Although join processing is a computation-intensive task, it frequently accesses neighbors of vertex $v$ following the adjacent edge label $l$, i.e., 
reading $N(v,l)$, where $N(v,l)=\{v^{\prime}| \overline{vv'}  \in E(G)  \wedge L_E (\overline{vv^{\prime}} ) = l\}$. 
To speed up this primitive, we propose a novel data structure to reduce memory latency. 
}

\emph{Compressed Sparse Row (CSR)} \cite{CSR} is widely used in existing algorithms (e.g., GunrockSM and GpSM) on sparse matrices or graphs, and it allows locating one vertex's neighbors in $O(1)$ time. 
Figure \ref{fig:oldcsr} shows an example: the 3-layer CSR structure of $G$ in Figure \ref{fig:example}.
The first layer is ``row offset'' array, recording the address of each vertex's neighbors.
The second layer is ``column index'' array, which stores all neighbor sets consecutively.
The corresponding weight/label of each edge is stored in ``edge value'' array.
If no edge weight/label exists, we can remove ``edge value'' array and yield 2-layer CSR structure.
\nop{
Traditional CSR cannot deal with edge labels. 
Figure \ref{fig:oldcsr} shows a simple way (Gunrock CSR structure, GCSR for short) to deal with edge labels, which is used in GunrockSM \cite{GunrockSM}. 
It adds an ``edge value'' array to record the edge label for each neighbor. 
In GCSR, when extracting $N(v,l)$, all neighbors of $v$ are accessed and checked whether or not corresponding edge label is $l$.
}
To extract $N(v,l)$ in CSR, all neighbors of $v$ must be accessed and checked whether or not corresponding edge label is $l$.
Obviously, the memory access latency is very high and it suffers from severe thread underutilization because threads extracting wrong labels are inactive thus wasted. 
We carefully design a GPU-friendly CSR variant to support accessing  $N(v,l)$ efficiently. 
\myhl{
The complexity of $N(v,l)$ extraction consists of locating and enumerating.
In our structures, $N(v,l)$ is stored consecutively, i.e., the complexity of enumerating is the same: $O(|N(v,l)|)$.
Thus, we use the time complexity of locating $N(v,l)$ as metric.
}

\begin{figure}[htbp]   
	\centering
\includegraphics[width=8cm]{\picfolder 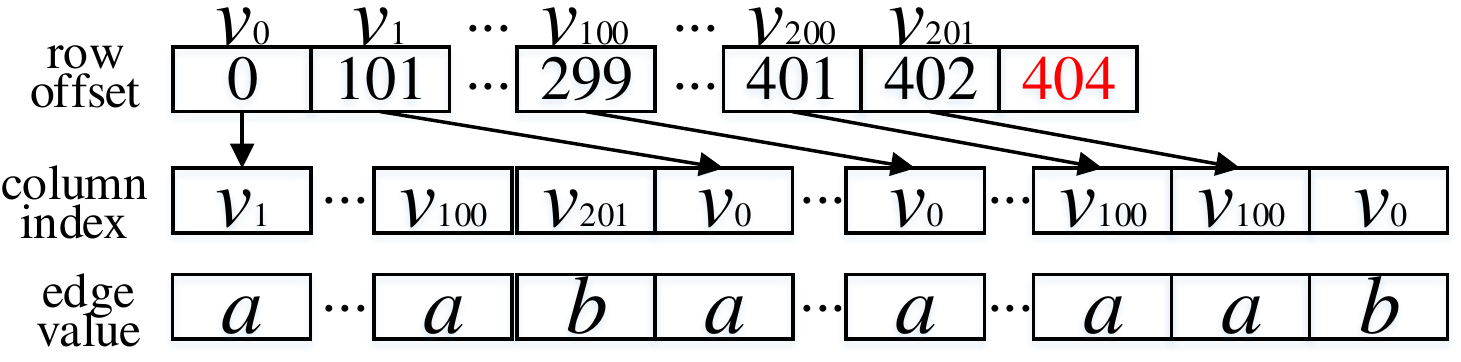}        	
	\vspace{-0.15in}     	
\caption{Traditional CSR structure}      	
\label{fig:oldcsr}    
\end{figure}

\nop{
\begin{definition}\label{def:partition}\textbf{Edge Label-Partitioned Graph.} 
    Given a graph $G$ and edge label $l$, the edge label-partitioned $P(G,l)$ of $G$ is the subgraph $G^{\prime}$ (of $G$) induced by all edges with label $l$.
\end{definition}
To speed up memory access, we divide graph $G$ into different \emph{edge label-partitioned graphs} (see Definition \ref{def:partition}). 
}

To speed up memory access, we divide  $G$ into different \emph{edge label-partitioned graphs} (for each edge label $l$, the edge $l$-partitioned $P(G,l)$ is the subgraph $G^{\prime}$ (of $G$) induced by all edges with label $l$).
These partitioned graphs are stored independently and edge labels are removed after partitioning.
The straightforward way is to store each one using traditional CSR. 
However, it cannot work well, since vertex IDs in a partitioned graph are not consecutive. 
For example, the edge partitioned graph $P(G,b)$ only has two edges and four vertices ($v_0$,$v_{1}$,$v_{101}$,$v_{201}$).  
The non-consecutive vertex IDs disable accessing the corresponding vertex in the row offest in $O(1)$ time (by vertex ID). 
There are two simple solutions:

\begin{figure*}
	\centering
        \includegraphics[width=18cm]{\picfolder 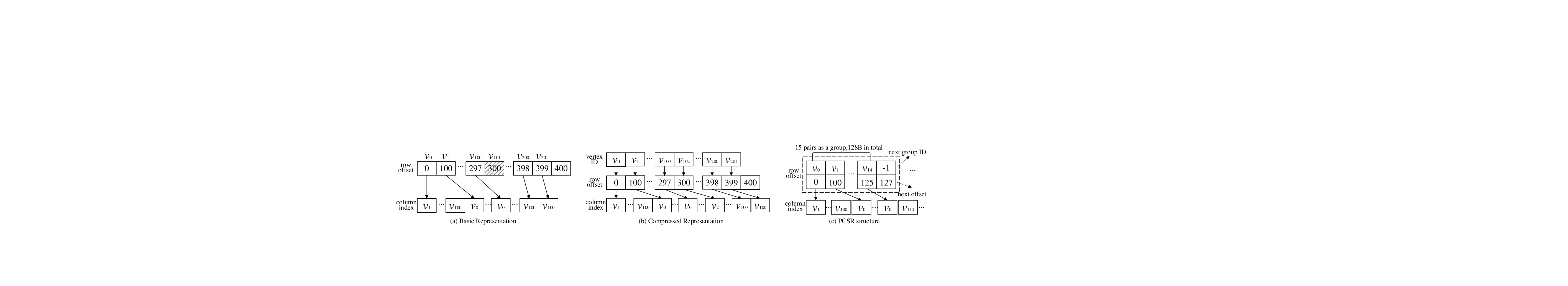}        	
    \vspace{-0.15in}
    \caption{Three Representations of edge $a$-partitioned graph}
    \label{fig:midcsr}      
\end{figure*}

(1) \emph{Basic Representation}. 
The entire vertex set $V(G)$ is maintained in the row offset for each edge partitioned graph CSR, regardless of whether or not a vertex $v$ is in the partitioned graph (see Figure \ref{fig:midcsr}(a)). 
Clearly, this approach can locate a vertex's neighbors in $O(1)$ time using the vertex ID directly, but it has high space cost: $O(|E(G)| + |L_{E(G)}|\times |V(G)|)$, where $ |L_{E(G)}|$ is the number of distinct edge labels. 
In complex graphs such as DBpedia, there are tens of thousands of different edge labels and this solution is not scalable.


(2) \emph{Compressed Representation}. 
A layer called ``vertex ID'' is added, and binary search is performed over this layer to find corresponding offset (see Figure \ref{fig:midcsr}(b)). 
Obviously, the overall space cost is lowered, which can be formulated as $O(|E(G)|)$. 
However, this leads to more memory latency. 
Theoretically, we require $\lceil \log{(|V(G,l)|+1)} \rceil+2$ memory transactions to locate $N(v,l)$, where $|V(G,l)|$ denotes the number of vertices in the edge $l$-partitioned graph $P(G,l)$. 

Therefore, neither of the above methods work for a large data graph $G$. 
In the following we propose a new GPU-friendly data structure to access $N(v,l)$ efficiently, called \emph{PCSR} (Definition \ref{def:pcsr}).
We reorganize the row offset layer using hashing.
The row offset layer is an array of hash buckets, called \emph{group}. 
Each item hashed to the group is a pair $(v, o_{v})$, where $v$ is a vertex ID and $o_{v}$ is the offset of $v$'s neighbors in column index $ci$. 
Let $GPN$ be a constant to denote the maximum number of pairs in each group. 
The last pair is an \emph{end flag} to deal with the overflow. 
We require that $2\leq GPN\leq 16$, then one group can be  read concurrently by a \emph{single} memory transaction using one warp.
\nop{
We require that $2\leq GPN\leq 16$, since a warp has 32 threads.
According to this structure, one group can be  read concurrently by a \emph{single} memory transaction using all threads within one warp.

In this way, 
We re-organize the row offset layer using a hash technique and give the data structure PCSR in Definition \ref{def:pcsr}.
Within a unique group, there are $B$ items combined into an array.
We require that $B$ is even and $B\leq 32$ so a group occupies no more than 128B memory and can be concurrently read by one memory transaction using a warp.
All PCSR structures are built independently and the process of building $PCSR(G,l)$ is presented in Algorithm \ref{alg:pcsr}.
In Line \ref{algcmd:overflow}, we call a group overflows if more than $\frac{B}{2}-1$ keys are mapped to it by $f$.
}

\begin{definition}\label{def:pcsr}\textbf{PCSR structure.} 
Given an edge $l$-partitioned graph $P(G,l)$, the Partitioned Compressed Sparse Row (PCSR for short) $PCSR(G,l)=\{gl, ci\}$ is defined as follows:
\begin{itemize}
\item $ci$ is the column index layer that holds the neighbors.
\item $gl=\{g_{i}\}$ is an array of groups and each group is a collection of pairs (no more than $GPN$ pairs). 
\item Each pair in $g_{i}$ is denoted as $(v, o_{v})$ except for the last pair, where $v$ is a vertex ID and $o_{v}$ is the offset of $v$'s neighbors in $ci$, i.e., a prefix sum of the number of neighbors for vertices. 
Let $n_{v}$ be the offset of next pair. 
$v$'s neighbors start at $ci[o_{v}]$ and end before $ci[n_{v}]$. 
All vertices in one group have the same hash value. 
\item The last pair $(GID,END)$ is the overflow flag. 
If $GID$ is -1, it means no overflow; otherwise, overflowed vertices are stored in the $GID$-th group. 
Note that $g_{i}.END$ is the end position of previous vertex's neighbors in $ci$, i.e., the first $o_{v}$ in group $g_{i+1}$.
\end{itemize}
\end{definition}

Figure \ref{fig:midcsr}(c) is an example of PCSR corresponding to edge $a$-partitioned graph. 
Let $D$ denote $P(G,l)$, the edge label $l$-partitioned graph. Algorithm \ref{alg:pcsr} builds PCSR for $D$. 
We allocate $|V(D)|$ groups (i.e. hash buckets) for $gl$ and $|E(D)|$ elements for $ci$ (Line \ref{algcmd:glalloc}). 
For each node $v$, we hash $v$ to one group using a hash function $f$ (Lines \ref{algcmd:hash1}-\ref{algcmd:hash2}). 
If some group $g_i$ overflows (i.e., more than $GPN-1$ vertices are hashed to this group), we find another empty group $g_j$ and record group ID of $g_j$ in the last pair in $g_i$ to form a linked list (Lines \ref{algcmd:regrp}-\ref{algcmd:gidll}). 
Claim \ref{clm:groups} confirms that we can always find empty groups to store these overflowed vertices. 
Finally, we put neighbors of each vertex in $ci$ consecutively and record their offsets in $gl$ (Lines \ref{algcmd:grpci}-\ref{algcmd:end}). 

\begin{claim} \label{clm:groups}
    When the overflow happens in Line \ref{algcmd:overflow} of Algorithm \ref{alg:pcsr}, we can always find enough empty groups to store all overflowed vertices.
\end{claim}
\optionshow{}{
\begin{IEEEproof} \label{prf:groups}
Firstly, once there is a hash conflict of keys, one more empty group arises.
Otherwise, let $x$ be the total number of conflicts and $y$ be the number of empty groups, we have $x\neq y$.
But it causes a paradox: $|V|-y$ is the number of non-empty groups and the number of all nodes should be equal to $|V|-y+x$, which means $y=x$. 
Secondly, $\forall g_{i}$, if group $g_{i}$ overflows, it needs to find $\lceil \frac{z}{GPN-1} \rceil -1$ empty groups where $z$ is the number of keys mapped to $g_{i}$.
The number of conflicts within group $g_{i}$ is $z-1$ and $g_{i}$ produces $z-1$ empty groups.
Obviously, $z-1\ge \lceil \frac{z}{GPN-1} \rceil -1$, so there are enough empty groups for each group's overflow.
These groups do not influence each other, thus the overall empty groups are enough.
\end{IEEEproof}
}

Based on PCSR, we compute one vertex's neighbors according to edge label. 
An example of computing $N(v_{0},a)$ in Figure \ref{fig:midcsr}(c) is given as follows. 
\begin{enumerate}
\item use the same hash function $f$ to compute the group ID $idx$ that $v_{0}$ maps to, here $idx=0$;
\item read the entire $0$-th group (i.e., $g_{0}$) to shared memory concurrently using one warp in one memory transaction;
\item probe all pairs $(v^{\prime},o_{v^{\prime}})$ in this group ($g_0$) concurrently using one warp;
\item we find the first pair (in group $g_0$) that contains $v_{0}$. The corresponding offset is 0 and the next offset 100. 
It means that $ci[0,...,99]$ in the column index layer are $v_{0}$'s neighbors. 
\end{enumerate}

Assume that vertex $v$ is hashed to the $i$-th group $g_i$. 
Due to the hash conflict, $v$ may not be in group $g_{i}$. 
In this case, according to the last pair, we can read another group whose ID is $g_{i}.GID$ and then try to find $v$ in that group. 
We iterate the above steps until $v$ is found in some group or a group is found whose $g_{i}.GID$ is ``-1'' (i.e., $v$ does not exist in $D$).


\emph{Parameter Setting}.
The choice of $GPN$ is critical to the performance of PCSR, affecting both time and space.
With smaller $GPN$, the space complexity is lower while the probability of group overflow is higher.
Once a group overflows, we may need to read more than one group when locating $N(v,l)$, which is more time consuming.
With larger $GPN$, the probability of group overflows is reduced, though the space cost rises.
Recall that the width of global memory transaction is exactly 128B, so in GSI we set $GPN=16$ to fully utilize  transactions.
Under this setting, there can be at most 15 keys within a group.
The space complexity is a bit high, which can be quantified as $32\times |V(D)|+|E(D)|$.
However, it is worthwhile and affordable because at any moment at most one partition is placed on GPU.
In addition, under this setting no group overflow occurs in any experiment of Section \ref{sec:experiment}.

\emph{Analysis}.
Within PCSR, $|V(D)|$ keys are hashed into $|V(D)|$ groups, which is called \emph{one-to-one hash} \cite{site:hash}. 
Under this condition, the time complexity of locating $N(v,l)$ can be analyzed by counting memory transactions. 
It is easy to conclude that the number of memory transactions is decided by the longest conflict list of one-to-one hash function. 
According to \cite{site:hash}, the expectation of longest conflict list's length is upper bounded by $1+\frac{5\log{|V(D)|}}{\log{\log{|V(D)|}}}$.
If $|V(D)| < 2^{32}$, the expectation of the maximum length of conflict list is smaller than 45. 
It means that at most $\left\lceil {\frac{{45}}{{GPN-1}}} \right\rceil  = \left\lceil {\frac{{45}}{{15}}} \right\rceil  = 3$ memory transactions are needed, since one transaction accesses one group and each group contain $GPN$ (=16) valid vertices. 
This is quite a large data graph.
In our experiments, even for graphs with tens of millions of nodes, the longest conflict list's length is no larger than 13, which means that only one memory transaction is needed to locate $N(v,l)$.
In other words, locating $N(v,l)$ is  $O(1)$. 
Furthermore, for each edge label $l$, the space cost of the corresponding PCSR is linear to $P(G,l)$. 
Thus, we can conclude that the total space of all PCSRs for $G$ is $O(|E(G)|)$. 
Table \ref{tab:csrs} summarizes the comparison, \myhl{where the time complexity counts both locating and enumerating $N(v,l)$ together}.


%



\begin{table}[htbp]
			\caption{\myhl{Efficiency of different data structures} } 	
	\begin{threeparttable}
\small
	\centering
    \begin{tabular}{ccc}
		\toprule
        Structure & Time Complexity & Space Complexity  \\
        \midrule
        CSR & $O(|N(v)|)$ & $O(|E|)$ \\
        \hline
        BR & $O(|N(v,l)|)$ & $O(|E|+|L_{E}|\times |V|)$     \\
        \hline
        CR & $O(\log{|V(G,l)|+|N(v,l)|})$ & $O(|E|)$    \\
        \hline
        PCSR &  $O(|N(v,l)|)$ & $O(|E|)$ \\
        \bottomrule
	\end{tabular}
\begin{tablenotes}
\item[*] {\small BR and CR denote ``Basic Representation'' and ``Compressed Representation'', respectively.}
\end{tablenotes}
	\label{tab:csrs}
	\end{threeparttable}
\end{table}

\nop{
Table \ref{table:csrs} compares efficiency of PCSR with the traditional CSR,  ``Basic Representation'' and ``Compressed Representation'' for graphs with edge labels.
It shows that PCSR can support computing $N(v,l)$  in amortized $O(1)$ time with the $O(|E(G)|)$ space, which is a practical and scalable data structure for $N(v,l)$ primitive.

While time complexity of PCSR is similar to the basic representation strategy(Figure \ref{fig:midcsr}(a)), the space complexity of \emph{PCSR} is also praiseworthy.
Notice that $|V(P(G,l))|\leq 2|E(P(G,l))|$, the complexity of $G$ can be formulated as 
$$\sum_{l \in L_{E(G)}}{32\times |V(P(G,l))|+E(P(G,l))}=O(\sum_{l \in L_{E(G)}}{|E(P(G,l))|})=O(|E(G)|)$$
The result shows that the space complexity of PCSR is similar to the compressed representation strategy(Figure \ref{fig:midcsr}(b)), which is independent to the number of edge labels in $G$.
Consequently, \emph{PCSR} fuses the merits of two strategies, which makes it a practical and scalable data structure for $N(v,l)$ extraction.
}

\nop{
\vspace{-0.15in}
\begin{figure}[htbp]   
	\centering
\includegraphics[width=8cm]{\picfolder 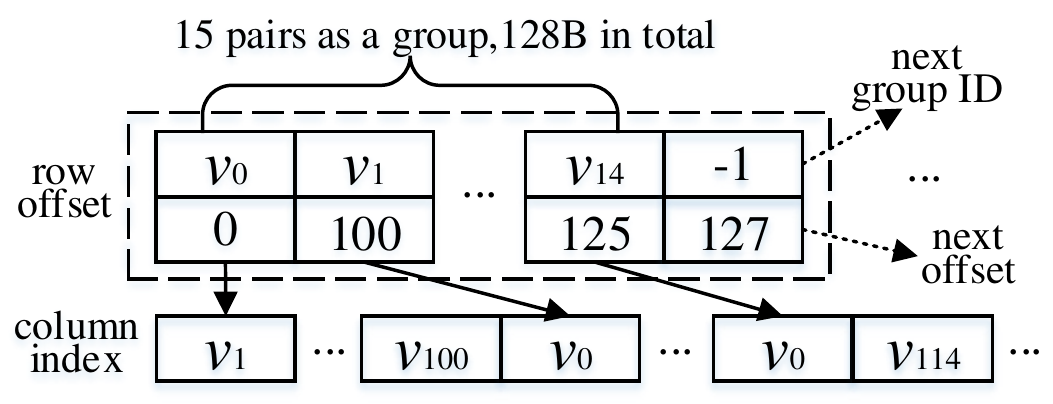}        	
	\vspace{-0.15in}     	
\caption{Edge $a$-partitioned PCSR structure}      	
\label{fig:pcsr}    
\end{figure}
\vspace{-0.15in}
}

%
\begin{algorithm}
	\small
\caption{Build PCSR structure}
\label{alg:pcsr}
\KwIn{partitioned graph $D=P(G,l)$}
\KwOut{PCSR structure of $D$}
allocate $gl$ array (containing $|V(D)$ groups) and $ci$ array (containing $|E(D)|$ elements); \label{algcmd:glalloc}  \\
select a hash function $f$, set $pos=0$;\\
\ForEach{node $v$ in $D$}
{ \label{algcmd:hash1}
    use $f$ to map $v$ to a group ID $i$; \label{algcmd:hash2}     \\
}
\ForEach{group $g_{i}$ in $gl$}
{ \label{algcmd:regrp}
\If{$g_{i}$ overflows}
{ \label{algcmd:overflow}
    find enough empty groups $g_{j}$ to store keys of $g_{i}$; \label{algcmd:groups}   \\
    set their $GID$s to form a linked list; \label{algcmd:gidll}   \\
}
}
\ForEach{group $g_{i}$ in $gl$}
{  \label{algcmd:grpci}
\ForEach{pair $T_{j}=\{v,o_{v}\}$ in $g_{i}$}
{
	set $o_{v}=pos$ in $T_{j}$;\\
	add $N(v)$ to $ci$ from $pos$ on and set $pos=pos+num(N(v))$;\\
}
set $END=pos$ in $g_{i}$; \label{algcmd:end}     \\
}
let $gl=\{g_{i}\}$ and return $\{gl,ci\}$ as data structure;\\
\end{algorithm}

%% file: join.tex
\section{Parallel Join Algorithm}\label{sec:basicjoin}

Algorithm \ref{alg:join} outlines the whole join algorithm, where the intermediate table $M$ stores all matches of partial query graph $Q^{\prime}$. 
In each iteration, we consider one query vertex $u$  and join intermediate table $M$ with candidate set $C(u)$ (Lines 9-11). 
Heuristically, the first selected vertex has the minimum score $score(u^{\prime})=\frac{C(u^{\prime})}{deg(u^{\prime})}$ (Lines 5-7). 
In later iterations, we consider the adjacent edge label frequency ($freq(l)$) when selecting the next query vertex to be joined (Lines 12-13). 

\begin{algorithm}
	\small
\caption{The whole join process}
\label{alg:join}
\KwIn{query graph $Q$, data graph $G$}
\KwOut{the final matches of $Q$ in $G$}
Let $Q^{\prime}$ be the partial query graph, set $Q^{\prime}=\phi$; \\
\ForEach{node $u^{\prime}$ in $Q$}
{
$score(u^{\prime})=\frac{C(u^{\prime})}{deg(u^{\prime})}$;\\
}
\For{$i=1$ to $|V(Q)|$}
{
\If{$i==1$}
{
$u_{c}=argmin_{u^{\prime}}{score(u^{\prime})}$; \\
set intermediate table $M=C(u_{c})$ and add $u_{c}$ to $Q^{\prime}$;  \\
}
\Else
{
$u=argmin_{u^{\prime}\notin Q^{\prime}}\{score(u^{\prime})|u^{\prime}$ is connected to $Q^{\prime}\}$;\\
Call Algorithm \ref{alg:joinTwo} to join $M$ with $C(u)$ (generating new intermediate table $M^{\prime}$);  \label{algcmd:calltwo}  \\
set $M=M'$, $u_{c}=u$ and add $u$ to $Q^{\prime}$;\\
}
\ForEach{edge $\overline{u_{c}u^{\prime}}$ in $Q$}
{
$score(u^{\prime})=score(u^{\prime})\times freq(L_{E}(\overline{u_{c}u^{\prime}}))$;\\
}
}
return $M$ as final result;\\
\end{algorithm}

Algorithm \ref{alg:joinTwo} lists how to process each join iteration (i.e., Line \ref{algcmd:calltwo} in Algorithm \ref{alg:join}). 
Before discussing the algorithm, we first study some of its key components. 
Each warp in GPU joins one row of $M$ with candidate set $C(u)$: acquires neighbors of vertices in this row leveraging restrictions on edge labels, and intersects them with $C(u)$ (\emph{set intersection}). 
The result of the intersection should remove the vertices in this row (\emph{set subtraction}), to satisfy the definition of isomorphism.

Let $Q^{\prime}$ be the partial query graph induced by query vertices $u_{0}$ and $u_{1}$.
Figure \ref{fig:newtable} shows the intermediate table $M$, in which each row $m_{i}$ represents a partial match of $Q^{\prime}$. 
Let  $L_{i}^{a}$ and $L_{i}^{b}$ be the neighbor lists of $m_{i}$, e.g., $L_{99}^{a}$ and $L_{99}^{b}$ represents $N(v_{0},a)$ and $N(v_{100},b)$ respectively.
For each row $m_{i}$, we assign a buffer ($buf_{i}$) to store temporary results.
Assuming that the next query vertex to be joined is $u_{2}$, let us consider the last warp $w_{99}$ that deals with the last row $m_{99}=\{v_{0},v_{100}\}$. 
There are two linking edges $\overline{u_0u_2}$ and $\overline{u_1u_2}$ with edge labels $a$ and $b$, respectively. 
Warp $w_{99}$ works as follows:
\begin{enumerate}
\item Read $v_0$'s neighbors with edge label $a$, i.e., $N(v_0,a)$;
\item Write $buf_{99}=(N(v_0,a) \setminus \{v_{0},v_{100}\}) \wedge C(u_2)$;
\item Read $v_{100}$'s neighbors with edge label $b$, i.e., $N(v_{100},b)$;
\item Update $buf_{99}=buf_{99} \wedge N(v_{100},b)$.
\item If $buf_{99} \neq \phi$, each item in $buf_{99}$ can be linked to the partial match $m_{99}$ to form a new match of $Q^{\prime} \cup u_{2}$. We write these matches to a new intermediate table $M^{\prime}$. 
\end{enumerate}
All warps execute the exact same steps as above in a massively parallel fashion on GPU, which will lead to some conflicts when accessing memory. 


\Paragraph{Problem of Parallelism}.
When all warps write their corresponding results to global memory concurrently, conflicts may occur. 
To enable concurrently outputting results, existing solutions use \emph{two-step output scheme}, which means the join is done twice. 
In the first round,  the valid join results for each warp are counted. 
Based on prefix-sum of these counts, each warp is assigned an offset. 
In the second round, the join process is repeated and join results are written to the corresponding addresses based on the allocated offsets. 
An example has been discussed in Example \ref{example:1}. 
Obviously, this approach doubles the amount of work.

\nop{
and concurrently running threads have no way to write results.
Traditional methods \cite{GpSM,GunrockSM} all use two-step output scheme to write the intermediate results to global memory, which means the joining phase are done twice:   
\begin{itemize}
\item At the first time, the whole join process is done and the matching number of each row is recorded.    
\item A prefix-sum operation is done on numbers of these valid matchings to generate the size and offsets of $M^{\prime}$.
\item Allocate space for $M^{\prime}$ on GPU memory.
\item At the second time, the whole join process is done and matching results of each row are written to corresponding offsets of $M^{\prime}$.
\end{itemize}
}

\Paragraph{Prealloc-Combine}. 
Our solution (Algorithm \ref{alg:joinTwo}) performs the join only once, which is called ``Prealloc-Combine''. 
Each warp $w_i$ joins one row ($m_i$) in $M$ with candidate set $C(u)$. 
Different from existing solutions, we propose ``Prealloc-Combine'' strategy. 
Before join processing, for each warp $w_i$, we allocate memory for $buf_{i}$ to store all valid vertices that can be joined with row $m_i$ (Line \ref{algcmd:memory} in Algorithm \ref{alg:joinTwo}). 
A question is how large this allocation should be.
Let $Q^{\prime}$ be the partial query graph that has been matched. 
We select one linking edge $e_{0}=\overline{u_{0}^{\prime}u}$ in query graph $Q$ ($u_0^{\prime} \in V(Q^{\prime})$), and $u$ ($\notin V(Q^{\prime})$) is the query vertex to be joined. 
Assume that the edge label is ``$l_{0}$''. 
As noted above, $m_i$ denotes one partial match of query graph $Q^{\prime}$. 
Assume that vertex $v^{\prime}_i$ matches $u_{0}^{\prime}$ in $m_i$. 
It is easy to prove that the capacity of $buf_{i}$ is upper bounded by the size of $N(v^{\prime}_i,l_{0})$. 
Based on this observation, we can pre-allocate  memory  of size $|N(v^{\prime}_i,l_{0})|$ for each row. 
Note that this pre-allocation strategy can only work for ``vertex-oriented'' join, since we cannot estimate the join result size for each row in the ``edge-oriented'' strategy. During each iteration, the selected edge $e_{0}$ is called \emph{the first edge} and it should be considered first in Line \ref{algcmd:lkedges} of Algorithm \ref{alg:joinTwo}. 
For example, in Figure \ref{fig:newtable}, $\overline{u_{1}u_{2}}$ is selected as $e_{0}$, thus the allocated size of $buf_{99}$ should be $|N(v_{100},a)|=3$.

\nop{
In ``vertex-oriented'' join strategy, the capacity of $buf_{i}$ can be upper bounded by the size of $N(v,l)$
Thanks to ``vertex-oriented'' join
in Algorithm \ref{alg:joinTwo} we store temporary results in buffers first, use prefix-sum to compute the size and offsets, and write temporary results to $M^{\prime}$ finally.
In this way, the join process is only done once, except for the small overhead of reading/writing temporary results.
However, another problem arises in Line \ref{algcmd:memory} because we need to prepare memory for $buf_{i}$ to store temporary results while the capacity is not known.
In the case of edge-by-edge joining, this question is insoluble; but in node-by-node joining, the capacity of $buf_{i}$ can be upper bounded by the size of $N(v,l)$.
Therefore, selecting an edge $e_{0}=(l_{0},u_{0})$, for each row $m_{i}$ of $M$, we can allocate in advance $num(N(m_{i}[u_{0}],l_{0}))$ memory as $buf_{i}$($m_{i}[u_{0}]$ is the mapping vertex of $u_{0}$ in $m_{i}$).
}

Though buffers can be pre-allocated separately for each row (i.e., each row issues a new memory allocation request), it is better to combine all buffers into a big array and assign consecutive memory space (denoted as $GBA$) for them (only one memory allocation request needed).
Each warp only needs to record the offset within $GBA$, rather than the pointer to $buf_{i}$. 
The benefits are two-fold:


(1) \emph{Space Cost}. Memory is organized as pages and some pages may contain a small amount of data.
In addition, pointers to $buf_{i}$ need an array for storage (each pointer needs 8B).
Combining buffers together helps reduce the space cost because it does not waste pages and only needs to record one pointer (8B) and an offset array (each offset only needs 4B). 

(2) \emph{Time Cost}. Combined preallocation has lower time overhead due to the reduction in the number of memory allocation requests.
Furthermore, the single pointer of $GBA$ can be well cached by GPU and the number of global memory load transactions decreases thanks to the reduction in the space cost of pointer array. \\

\begin{algorithm}
    \small
	\caption{Join a new candidate set}
	\label{alg:joinTwo}
	\KwIn{query graph $Q$, current intermediate table $M$ corresponding to the partial matched query $Q^{\prime}$, candidate set $C(u)$ ($u$ is the vertex to be joined), and linking edges $ES$ between $Q^{\prime}$ and $u$.}
	\KwOut{updated intermediate table $M^{\prime}$}
	Call Algorithm \ref{alg:preallocate} to select the first edge $e_{0}$, and pre-allocate memory $GBA$ and offset array $F$.  \label{algcmd:memory} \\
	\ForEach {linking edge $e=\overline{u^{\prime}u}$ in $ES$}
	{  \label{algcmd:lkedges}
		let $l$ be the label of edge $e$ in $Q$; \\
        launch a GPU kernel function to join $M$ with $C(u)$ \label{algcmd:join};\\
		\ForAll{each row $m_{i}$ (partial match) in $M$}
		{
			let $buf_{i}$ be the segment $F_{i}$\textasciitilde$F_{i+1}$ in $GBA$; \\ 
			assign a unique warp $w_{i}$ to deal with $m_{i}$;\\
			assume that $v_{i}^{\prime}$ match $u^{\prime}$ in $m_i$; \\ 
            \If{$e$ is the first edge $e_{0}$}  
			{ \label{algcmd:first}
				do set subtraction $buf_{i}=N(v_{i}^{\prime},l) \setminus m_{i}$ \label{algcmd:subtract};\\
				do set intersection $buf_{i}=buf_{i} \cap C(u)$ \label{algcmd:intersect};\\		
			}
			\Else
			{
				do set intersection $buf_{i}=buf_{i} \cap N(v_{i}^{\prime},l)$ \label{algcmd:merge};\\
			}
		}
	}
    do prefix-sum scan on $\{num(buf_{i})\}$;  \label{algcmd:newscan}   \\
	allocate memory for new intermediate table $M'$; \label{algcmd:newalloc}  \\
	launch a GPU kernel function to link $M$ and $buf_{0,...|M|-1}$ to generate $M'$ \label{algcmd:link};\\
	\ForAll{partial answer $m_{i}$ in $M$}
	{
		read $m_{i}$ into shared memory;\\
		assign a unique warp $w_{i}$ to deal with $m_{i}$;\\
		\ForAll{$z$ in $buf_{i}$}
		{
			copy $m_{i}$ and $z$ to the corresponding address of $M'$ as a new row; \label{algcmd:newtable}  \\
		}
	}
	return $M'$ as the result;\\
\end{algorithm}

Algorithm \ref{alg:preallocate} shows how to allocate buffers $buf_i$ for each row $m_i$. 
Assume that there exist multiple linking edges between $Q^{\prime}$ (the matched partial query graph) and vertex $u$ (to be joined). 
To reduce the size of $|GBA|$, among all linking edges, we select the linking edge $\overline{u_{0}^\prime u}$ whose edge label $l_{0}$ has the minimum frequency in $G$ (Line \ref{algcmd:minedge}).
We perform a parallel exclusive prefix-sum scan on each row's upper bound $|N(v_{i}^{\prime},l_{0})|$ (Lines \ref{algcmd:scan0}-\ref{algcmd:scan}), later the offsets ($F[i]$, $\forall 0\leq i<|M|$) and capacity of $GBA$ ($F[|M|]$) are acquired immediately. 
With the computed capacity, we pre-allocate the $GBA$ and offset array $F[0,...,|M|-1]$ (Line \ref{algcmd:GBAalloc}). 
Each buffer $buf_i$ begins with the offset $F[i]$. 

Let us recall Figure \ref{fig:newtable}, where Figure \ref{fig:newtable}(a) is the process of $GBA$ allocation.
First, a parallel exclusive prefix sum is done on $num(L_{i}^{a})$ and the size of $GBA$ is computed (200).
Then $GBA$ is allocated in global memory and the address of $buf_{i}$ is acquired.
For example, the final row $m_{99}$ has three edges labeled by $a$, thus $num(L_{99}^{a})$ is 3 and the beginning address of $buf_{99}$ in $GBA$ is 197.
However, if $\overline{u_{0}u_{2}}$  is selected as the first edge $e_{0}$, we can yield smaller $|GBA|$ (100).
The label $b$ of $\overline{u_{0}u_{2}}$ is more infrequent than $a$, thus heuristically it is superior, as illustrated in Algorithm \ref{alg:preallocate}.
For ease of presentation, we still assume that $\overline{u_{1}u_{2}}$ is selected as $e_{0}$ in Figure \ref{fig:newtable}.

\nop{
Each offest 

$buf_{i}$

for all buffers and each row $m_{i}$ can find $buf_{i}$ by the array pointer and corresponding offset. 
}

In each join iteration, Algorithm \ref{alg:joinTwo} handles all linking edges between $Q^{\prime}$ and $u$.
It allocates $GBA$ (Line \ref{algcmd:memory}), processes linking edges one by one (Lines \ref{algcmd:lkedges}-\ref{algcmd:merge}), and finally generates a new intermediate table $M^{\prime}$(Lines \ref{algcmd:newscan}-\ref{algcmd:newtable}).
Obviously, $GBA$ is allocated only once in Algorithm \ref{alg:joinTwo} and no new temporary buffer is needed.
Figure \ref{fig:newtable}(a) performs the $GBA$ allocation by edge $\overline{u_{1}u_{2}}$ and Figure \ref{fig:newtable}(b) finishes set operations.
Correspondingly, edge $\overline{u_{1}u_{2}}$ is joined first.
For example, $L_{99}^{a}$ subtracts $m_{99}$ and the result is $\{v_{200},v_{201}\}$, which are stored in $buf_{99}$ (Line \ref{algcmd:subtract}).
Next, for each valid element $x$ in $buf_{99}$, we check its existence in candidate set of $u_{2}$ (Line \ref{algcmd:intersect}).
The second edge is $\overline{u_{0}u_{2}}$ and it is processed by Line \ref{algcmd:merge}, where $buf_{99}$ is further intersected with $L_{99}^{b}$ and the result is $\{v_{201}\}$, i.e., $num(buf_{99})=1$.
We acquire the matching vertices of each row $m_{i}$ in $buf_{i}$, then a new prefix sum is performed to obtain size and offsets of $M^{\prime}$ (Line \ref{algcmd:newscan}).
After $M^{\prime}$ is allocated, $w_{i}$ copies extensions of $m_{i}$ to $M^{\prime}$ (Lines \ref{algcmd:newalloc}-\ref{algcmd:newtable}).

\nop{
\begin{claim}\label{clm:buffer}
The number of valid elements in $buf_{i}$ never exceeds the capacity.
\end{claim}
\begin{proof}\label{prf:buffer}
The first edge $e_{0}$ is processed first and current number of valid elements is the same as $|buf_{i}|$.
Line \ref{algcmd:subtract} does set subtraction and Line \ref{algcmd:intersect} does set intersection.
They donot add valid elements to $buf_{i}$.
Neither does Line \ref{algcmd:merge}, which processes other edges.
\end{proof}
}

\begin{algorithm}
	\small
	\caption{Function: Pre-allocate Memory}
	\label{alg:preallocate}
	\KwIn{query graph $Q$, current intermediate table $M$ corresponding to the partial matched query $Q^{\prime}$, candidate set $C(u)$ ($u$ is the query vertex to be joined), and linking edges $ES$ between $Q^{\prime}$ and $u$.}
	\KwOut{Allocated memory $GBA$ and Offset arrary $F$.}
    Among all edges in $ES$, select edge $e_{0}=\overline{u_{0}^{\prime}u}$, whose edge label $l_{0}$ has the minimum frequency in $G$.  \label{algcmd:minedge} \\
	Set offset $F[0]$=0; \\
	\ForEach{row $m_i$ in $M$, $i=0,...,|M|-1$}
	{ \label{algcmd:scan0}
		Assume vertex $v_{i}^{\prime}$ matches query vertex $u_{0}^{\prime}$ in row $m_i$. \\
	   $F[i+1]$=$F(i)$+$|N(v_{i}^{\prime},l_{0})|$.  // Do exclusive prefix-sum scan.    \label{algcmd:scan}  \\
	}
    Let $|GBA|=F[|M|]$;  \\
    Allocate consecutive memory with size $|GBA|$ and let $GBA$ record the beginning address. \label{algcmd:GBAalloc}   \\
	Return $GBA$ and offset array $F[0,...,|M|-1]$. 	
\end{algorithm}

\Paragraph{GPU-friendly Set Operation}.
In Algorithm \ref{alg:joinTwo}, set operations (Lines \ref{algcmd:subtract},\ref{algcmd:intersect},\ref{algcmd:merge}) are in the innermost loop, thus frequently performed.
Traditional methods (e.g., \cite{fox2018fast}) all target the intersection of two lists. 
However, in our case there are many lists of different granularity for set operations.
A naive implementation launches a new kernel function for each set operation and uses traditional methods to solve it.
This method performs bad, so we propose a new GPU-friendly solution.

There are three granularities: small (partial match $M_{i}$), medium (neighbor list $N(v,l)$)  and large (candidate set $C(u)$).
We use one warp for each row and design different strategies for these lists:
\begin{itemize}
    \item For small list $M_{i}$, we cache it on shared memory until the subtraction finishes.
    \item For medium list $N(v,l)$, we read it batch-by-batch (each \emph{batch} is 128B) and cache it in shared memory, to minimize memory transactions.
    \item For large list $C(u)$, we first transform it into a bitset, then use exactly one memory transaction to check if vertex $v$ belongs to $C(u)$.
\end{itemize}
Lines \ref{algcmd:subtract} and \ref{algcmd:intersect} can be combined together.
After subtraction, the check in Line \ref{algcmd:intersect} is performed on the fly.

We also add a write cache to save write transactions, as there are enormous invalid intermediate results which do not need to be written back to $buf_{i}$.
It is exactly 128B for each warp and implemented by shared memory.
Valid elements are added to cache first instead of written to global memory directly.
Only when it is full, the warp flushes its cached content to global memory using exactly one memory transaction.

%% file: optimization.tex
\section{Optimizations}  \label{sec:optimize}

\optionshow{}{
There are two more optimizations in Algorithm \ref{alg:joinTwo}: improving workload balance and elimination of duplicate vertices.
We  discuss these below.
}

\subsection{Load Balance}\label{sec:balance}

In Algorithm \ref{alg:joinTwo}, load imbalance mainly occurs in Lines \ref{algcmd:join} and \ref{algcmd:link}, where neighbor set sizes of all rows are distributed without attention to balance.
We propose to balance the workload using the following method (\emph{4-layer balance scheme}): (1) Extract workloads that exceed $W_{1}$, and dynamically launch a new kernel function to handle each one; (2) Control the entire block to deal with all workloads larger than $W_{2}$; (3) In each block, all warps add their tasks  exceeding $W_{3}$ to shared memory and then divide them equally; and (4) Each warp finishes remaining tasks of the corresponding row.

The first strategy limits inter-block imbalance; the next two limit imbalance between warps.
$W_{2}$ should be set as the block size of CUDA, while $W_{1}$ and $W_{3}$ are parameters that should be tuned ($W_{1}>W_{2}>W_{3}>32$).
This method is superior to merging all tasks and dividing them equally \cite{DBLP:conf/ppopp/MerrillGG12}, because it avoids the overhead  of merging tasks into work pool.

\nop{
In the first layer, $W_{1}$ is a parameter which is much larger than the block size.
If workload of some row exceeds $W_{1}$, it does not end even after all other tasks are finished.
In this case, this task can only utilize a single SM while all other SMs are idle.
Therefore, a new kernel is launched to finish this kind of tasks first, which uses larger block or more blocks for each task.
After the first layer, warps of a block vote to control the whole block to deal with workloads between $W_{2}$ and $W_{1}$.
$W_{2}$ is a parameter larger than or equal to the block size, for example, it can be set as twice the block size.
Later, remaining workloads are all below $W_{2}$.
In the third layer, for each block, tasks exall warps.ceeding $W_{3}$ are added into a work pool in shared memory first and then evenly distributed to 
$W_{3}$ is a parameter larger than the warp size, for example, it can be set as twice the warp size.
With the two strategies above, load imbalance between warps is bounded by $W_{3}$.
Finally, each warp finish the corresponding workload whose size is below $W_{3}$.
Instead of merging all tasks and dividing them equally (see \cite{DBLP:conf/ppopp/MerrillGG12}), the 4-layer balance scheme is adopted because it is time-consuming to add tasks to work pool.
}

\subsection{Duplicate Removal}\label{sec:dupr}

In Figure \ref{fig:newtable}, the first elements of all rows are all $v_{0}$ and each row does the same operation: extracting $N(v_{0},a)$.
To reduce redundant memory access, we propose a heuristic method to remove duplicates within a block.
If rows $x$ and  $y$ have a common vertex $v$ in the same column, we let the two warps of $x$ and $y$ ($wx$ and $wy$) share the input buffer (placed in shared memory) of $N(v,l)$.
For the shared input buffer, only a single warp (e.g., warp $wx$) reads neighbors into buffer.
Other warps wait for the input operation to finish and then all warps perform their own operations.
\optionshow{The pseudo code is in \cite{fullVersion}.}{}
\optionshow{}{
Algorithm \ref{alg:dupr} gives implementation details.
\begin{algorithm}
	\small
	\caption{Duplicate removal within each block}
	\label{alg:dupr}
	\KwIn{vertex $v_{i}$ and input buffer $buf_{i}$ for each warp $w_{i}$}
	\ForEach{warp $w_{i}$ in the block}
	{
		$id[i]$ = $v_{i}$;\\
		synchronize all warps within the block;\\
		use $w_{i}$ to find the first occurrence $j$ of $v_{i}$ in $id[]$;\\
		$addr[i] = j$;\\
	}
	\ForEach{batch $b_{i}$ of $buf_{i}$}
	{
		\If{$addr[i]==i$}
		{
			use $w_{i}$ to read batch $b_{i}$ into buffer $buf_{i}$;\\
		}
	    synchronize all warps within the block;\\
	    each warp $w_{i}$ processes the batch located in $buf_{addr[i]}$;\\
	}
\end{algorithm}
}

%% file: extension.tex
\section{Extensions}  \label{sec:extension}

\subsection{Homomorphism and edge isomorphism}

Our method can be generalized to support homomorphism and edge isomorphism subgraph query semantics.
Both of these semantics are equally important as vertex isomorphism and to all foresight will be supported by upcoming query languages standards.
We briefly discuss them below. 

\begin{figure}[htbp]
	\centering		
	\includegraphics[width=8cm]{\picfolder 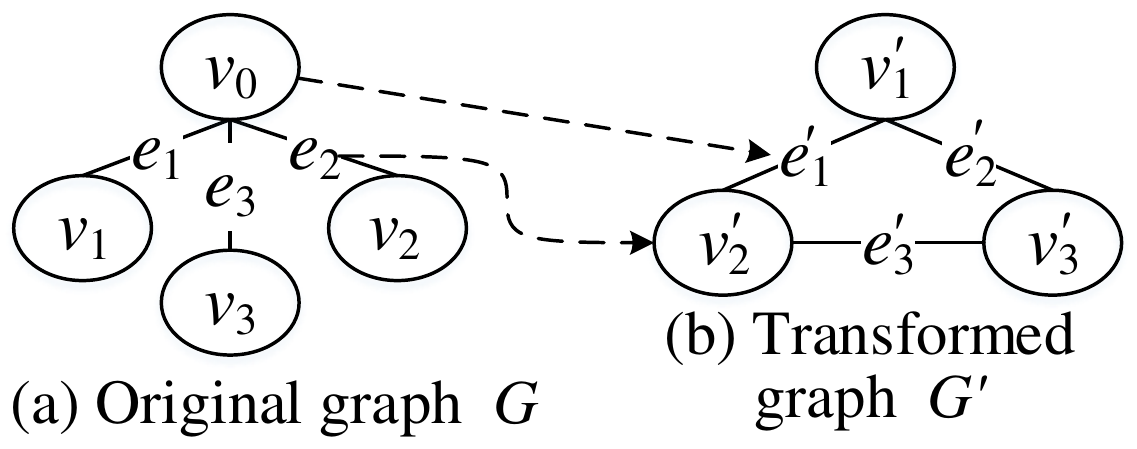}	
	\vspace{-0.15in}
	\caption{Example of edge isomorphism}	
	\label{fig:edgeiso}
\end{figure}


According to \cite{DBLP:journals/ijprai/ConteFSV04}, \emph{homomorphism} semantics drops the constraint in subgraph isomorphism that two different query vertices in $Q$ should be matched to two distinct vertices in the data graph $G$. 
Therefore, to support homomorphism, we only need to eliminate the set subtraction operation (Line 10 in Algorithm 3) in our joining phase. 

According to \cite{gross2004handbook}, the semantics of edge isomorphism requires that two edges (in $Q$) share a common vertex of $Q$ if and only if the mapping edges (in $G$) share a common vertex of $G$.
To support it, we can transform the original graph $G$ into a new graph $G^{\prime}$. 
Each edge $e$ of $G$ is transformed into a vertex $v^{\prime}$  of $G^{\prime}$, and each vertex $v$ is transformed into several edges of $G^{\prime}$. 
Figure \ref{fig:edgeiso} shows an example.
Let $e_{1}=\overline{v_{0}v_{1}}$, $e_{2}=\overline{v_{0}v_{2}}$ and $e_{3}=\overline{v_{0}v_{3}}$ be three edges in $G$; we transform $e_{1}$, $e_{2}$ and $e_{3}$ into three vertices $v_{1}^{\prime}$, $v_{2}^{\prime}$ and $v_{3}^{\prime}$ of $G^{\prime}$.
Since $e_{1}$ and $e_{2}$ share the common vertex $v_{0}$, we transform $v_{0}$ into an edge $e_{1}^{\prime}=\overline{v_{1}^{\prime}v_{2}^{\prime}}$ of $G^{\prime}$.
The same thing happens between $e_{1}$ and $e_{3}$, $e_{2}$ and $e_{3}$, which means $v_{0}$ is transformed into $e_{1}^{\prime}$, $e_{2}^{\prime}$ and $e_{3}^{\prime}$.
Such transformation is also performed on query graph $Q$ to generate $Q^{\prime}$.
Then, we can perform subgraph isomorphism of $Q^{\prime}$ on $G^{\prime}$ to get results $R^{\prime}$.
The final results $R$ of searching $Q$ on $G$ can be acquired via reverse transformation of $R^{\prime}$.

\subsection{Process vertex/edge multi-labels}

In some situations, a vertex (or an edge) may have multiple labels.
For vertex/edge multi-labels, we should first change the definition of subgraph isomorphism. 
A proper change is: For each vertex $u$/edge $\overline{u_1u_2}$ in query graph $Q$, assume that the matching vertex/edge is denoted as $f(u)$ and $f(\overline{u_1u_2})$, respectively and we require that $L_{V}(u) \subseteq L_{V}(f(u))$ and $L_{E}(\overline{u_1u_2}) \subseteq L_{E}(\overline{f(u_1)f(u_2)})$. 

Under the above definition,  the data structures and algorithms in both filtering and join phases in GSI should be adjusted as follows.

(1) For vertices with multiple labels, we only need to modify the filtering phase. In the case of single-label vertex $v$, we use a special strategy to process vertex label, which stores the label directly in the beginning of $v$'s signature (see Section \ref{sec:expFilter}. 
However, if $v$ has multiple labels, this special strategy can not be used, i.e., we need to use hash functions for vertex labels.
Given a data vertex $v$ (or query vertex $u$), we can hash all labels of $v$ (or $u$) into its signature, then perform the original filtering phase.
The main difference is that candidate sets must be refined by verifying the containment of vertex labels, which can be finished quickly on GPU.
Given a query vertex $u$ and its candidate set $C(u)$, each thread can examine each candidate $v$ in $C(u)$, checking if the label set of $v$ contains the label set of $u$. After the refinement, the joining phase can work properly because it does not consider vertex labels.

(2) For edges with multiple labels, we work as follows. 
Given graphs $Q$ and $G$ with multiple edge labels, we transform them as multiple-edge graphs $Q^{\prime}$ and $G^{\prime}$, respectively, where each edge has a single edge label  (as shown in the following Figure \ref{fig:multilab}). 
The whole GSI algorithm does not need be changed to handle the multiple-edge graph, as GSI always processes one edge at a time (Line 3-13 of Algorithm 3). 

\begin{figure}[htbp]
	\centering		
	\includegraphics[width=8cm]{\picfolder 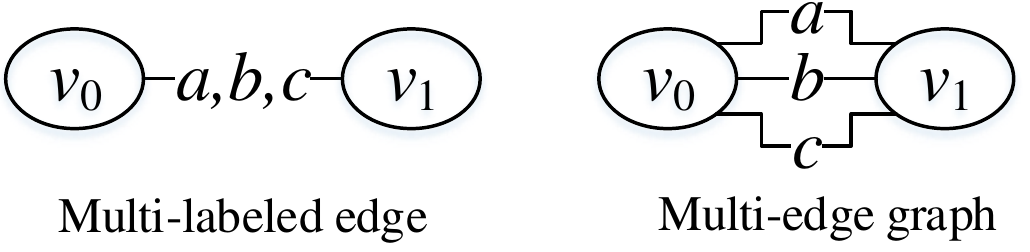}	
	\vspace{-0.15in}
	\caption{Example of multi-labeled edge}	
	\label{fig:multilab}
\end{figure}

%% file: experiment.tex
\section{Experiments}\label{sec:experiment}

In this section, we evaluate our method (GSI) against state-of-the-art subgraph matching algorithms, such as CPU-based solutions VF3 \cite{VF3}, CFL-Match \cite{CFL-Match}, \myhl{CBWJ \cite{CBWJ}}, and GPU-based solutions GpSM and GunrockSM. 
\myhl{
We also include two state-of-the-art GPU-based RDF systems (MAGiQ \cite{MAGiQ} and Wukong+G \cite{Wukong+G}) in the experiments.
Note that RDF systems are originally designed for SPARQL queries whose semantic is subgraph homomorphism; we extend them to support subgraph isomorphism.
}
All experiments are carried out on a workstation running CentOS 7 and equipped with Intel Xeon E5-2697 2.30GHz CPU and 188G host memory, NVIDIA Titan XP with 30 SMs (each SM has 128 cores and 48KB shared memory) and 12GB global memory. 

\subsection{Datasets and Queries} \label{sec:dataset}
The experiments are conducted on both real and synthetic datasets. 
The statistics are listed in Table \ref{tab:dataset}.  
Enron email communication network (\emph{enron}), the Gowalla location-based social network (\emph{gowalla}), \myhl{patent citation network (\emph{patent})} and the road\_central USA road network (\emph{road}) are downloaded from SNAP \cite{site:snap}. 
Large RDF graphs, such as DBpedia \cite{dbpedia} and WatDiv (a synthetic RDF benchmark \cite{watdiv}), are also used. 

\vspace{-0.2in}
\begin{table}[htbp]
	\caption{Statistics of Datasets}
	\vspace{-0.1in}
	\begin{threeparttable}
		\small
		\centering
		\begin{tabular}{crrrrrr}
			\toprule
			Name & $|V|$ & $|E|$ & $|L_{V}|$ & $|L_{E}|$ & MD\tnote{1} & Type\tnote{2} \\
			\midrule
			enron & 69K & 274K & 10 & 100 & 1.7K & rs \\
			gowalla & 196K & 1.9M & 100 & 100 & 29K & rs \\
			\myhl{patent} & \myhl{6M} & \myhl{16M} & \myhl{453} & \myhl{1K} & \myhl{793} & \myhl{rs} \\
			road & 14M & 16M & 1K & 1K & 8 & rm \\
			DBpedia & 22M & 170M & 1K & 57K & 2.2M & rs \\
			WatDiv & 10M & 109M & 1K & 86 & 671K & s \\
			\bottomrule
		\end{tabular}
		\begin{tablenotes}
			\item[*] $|L_{V}|$ and $|L_{E}|$ denote the number of vertex label and edge label, respectively.
			\item[1] Maximum degree of the graph.
			\item[2] Graph type: r:real-world, s:scale-free, and m:mesh-like.
		\end{tablenotes}
		\label{tab:dataset}
	\end{threeparttable}
\end{table}
\vspace{-0.15in}

Since most graphs do not contain vertex/edge labels except for edge labels in RDF datasets and \myhl{vertex labels in patent dataset}, we assign labels following the power-low distribution. 
The default numbers of vertex/edge labels are given in Table \ref{tab:dataset}. 
To generate a query graph, we perform the random walk over the data graph $G$ starting from a randomly selected vertex until $|V(Q)|$ vertices are visited. 
All visited vertices and edges (including the labels) form a query graph. 
The same query graph generation approaches are also used in \cite{yan:gIndex,DBLP:journals/pvldb/HanLPY10}. 

For each query size $|V(Q)|$, we generate 100 query graphs and report the average query running time. 
Note that the default query size $|V(Q)|$ is 12 in the following experiments. 
In Section \ref{sec:lqsize},  we also evaluate  GSI with respect to the number of vertex/edge labels and query size. 

\nop{
we follow a 20/80 distribution which is more realistic.

is also used in our experiments. Synthetic datasets are generated by WatDiv RDF benchmark \cite{watdiv} and transformed into the graphs we need. 
Our datasets cover graphs from different areas with different sizes and distributions, providing a representative benchmark.
The real-world data come from SNAP \cite{site:snap}, which include the Enron email communication network, the Gowalla location-based social network and the road\_central USA road network.  
Besides, DBpedia170M RDF dataset \cite{dbpedia} is also transformed and tested in experiments.
Synthetic datasets are generated by WatDiv RDF benchmark \cite{watdiv} and transformed into the graphs we need.
Our datasets cover graphs from different areas with different sizes and distributions, providing a representative benchmark.
\emph{Label Generation}.
As for labels, we set different numbers of labels for different datasets.
For large datasets, we use more labels, while for small datasets we use less.
RDF datasets already have edge labels, so we just generate veretx labels for them.
For other datasets, we generate both vertex labels and edge labels.
Instead of generating labels randomly, we follow a 20/80 distribution which is more realistic.
For example, in the Enron dataset we use 10 labels for vertices, where 80\% vertices have label 1(40\%) or 2(40\%).
\emph{Query Generation}.
We require that queries are randomly generated and have valid matches in the datasets.
Therefore, each query is generated by a random walk within the data graph.
The number of vertices in each query is restricted to 3\textasciitilde12.
For each dataset, 100 queries are generated by randomly digging from the data graph.
Each query $Q$ is run 3 times and the average is selected as the responding time of $Q$.
Finally, the responding time of 100 queries is averaged to be the metric of performance on a dataset.
}

\nop{
NVIDIA Titan XP is installed on this server, which has 30 SMs(each SM has 128 cores and 48KB shared memory) and 12GB global memory.

The performance of GSI is evaluated in comparison with state-of-the-art subgraph matching algorithms, including VF2 \cite{DBLP:journals/jacm/Ullmann76}, TurboISO \cite{DBLP:conf/sigmod/HanLL13}, GpSM and GunrockSM.
VF2 and TurboISO are widely used CPU algorithms based on backtracking, while GpSM \cite{GpSM} and GunrockSM \cite{GunrockSM} are GPU algorithms.
None of them gives exactly the implementation in our cases, so we try our best to implement the four counterparts.

GSI is implemented using \emph{CUDA 8.0} and \emph{CUB} library \cite{site:cub} and compiled using \emph{NVCC} with flag {\itshape -arch=sm\_35} and {\itshape -O2}.
In all experiments, $N$ and $K$ are set as 512 and 32 respectively in the filtering phase.
Murmurhash 2.0 is used as the hash function in our implementation.
In \emph{Load Balance} scheme of Section \ref{sec:optimization}, we set $W_{1}=4096$, $W_{2}=1024$ and $W_{3}=128$ with the block size and warp size are 1024 and 32 respectively.
}

\nop{
\emph{Environment}.
All experiments are carried out on a server with CentOS 7 installed.
The server is equipped with Intel Xeon E5-2697 2.30GHz CPU and 188G memory.
NVIDIA Titan XP is installed on this server, which has 30 SMs(each SM has 128 cores and 48KB shared memory) and 12GB global memory.
}

\nop{
\emph{Measure}.
Data graph is loaded into memory first and preprocessed to build encoding table and PCSR structures, after which the queries are read one by one and executed.
For each query, we begin counting the time after the query graph is read into memory and end when all results are flushed to a disk file.
We report error if some algorithm fails to answer a query within 100 seconds.
Time of building encoding table and PCSR structures is not listed due to limit of pages.
On the other hand, preprocessing is done offline and consumes several hundreds of seconds at most in our experiments.
}

\optionshow{}{
\subsection{Evaluating Filtering Strategy} \label{sec:expFilter}
Let us recall the encoding technique in Section \ref{sec:filter}. 
The neighborhood structure around each vertex $v$ is encoded into a length-$N$ bitvector signature $S(v)$. 
Furthermore, $K$ of $N$ bits denote vertex $v$'s label and the left bits correspond to $v$'s adjacent edges and neighbors (an example is given in Figure \ref{fig:encodeTABLE}(a)). 
In our experiments, we set $N$=512 and $K$=32. 
By varying the encoding length, we can balance the filtering time and the pruning power.
\optionshow{
The tuning of $N$ and $K$ is  in the full version of the paper \cite{fullVersion}.
}{}

To verify the effectiveness of our encoding, we compare it with the pruning techniques (used in GpSM and GunrockSM) that are based on node label and degree. 
The metrics include time cost and the size of the minimum candidate set, because the joining phase always begins from the minimum candidate set. 
Experimental results (Table \ref{tab:filterPerf}) show that our encoding strategy not only obtains much smaller candidate sizes (reduces 10-100 times) than the filtering in existing algorithms but also consumes less pruning time. 
The superiority of our filtering method is due to the careful design of signature structure on GPU.
Natural load balance is achieved, and the column-first layout of vertex signatures also brings performance improvement due to coalesced memory access.

\begin{table}[htbp]
			\caption{Performance of different filtering strategies}
		\vspace{-0.1in}
	\begin{threeparttable}
		\small
		\centering
		\begin{tabular}{|c|r|r|r|r|r|r|}
			\hline
			\multirow{2}*{Dataset} & \multicolumn{3}{|c|}{Minimum $|C(u)|$} & \multicolumn{3}{|c|}{Time (ms)} \\ 
			\cline{2-7}
			~ &  GpSM & GSM\tnote{1} & GSI &  GpSM & GSM & GSI   \\ 
			\hline
			enron & 2,246 & 2,270 & 111 &  24 & 20 & 9  \\
			\hline
			gowalla &  153 & 1,072 & 90 &  31 & 24 & 16  \\
			\hline
			patent & 2,298 & 3,401 & 25 & 781 & 569 & 354   \\
			\hline
			road & 8 & 2,544 & 7 &  259 & 394 & 187   \\
			\hline
			WatDiv &  871 & 12,145 & 604 & 290 & 252 & 201   \\
			\hline
			DBpedia &  138 & 11,405 & 132 &  410 & 494 & 407    \\
			\hline
		\end{tabular}
		\begin{tablenotes}
			\item[1] The filtering strategy of GunrockSM.
		\end{tablenotes}
		\label{tab:filterPerf}
	\end{threeparttable}
	\vspace{-0.15in}
\end{table}
}

\begin{table*}[htbp]
			\caption{Performance of techniques in join phase}
			\vspace{-0.1in}
	\begin{threeparttable}
		\small
		\centering
		\begin{tabular}{|c|r|r|r|r|r|r|r|r|r|r|r|r|r|r|}
			\hline
			\multirow{2}*{Dataset} & \multicolumn{7}{|c|}{Global Memory Load Transactions} & \multicolumn{7}{|c|}{Query Response Time (ms)} \\ 
			\cline{2-15}
			~ & GSI-\tnote{1} & +DS\tnote{2} & drop & +PC\tnote{3} & drop & +SO\tnote{4} & drop & GSI- & +DS & speedup & +PC & speedup & +SO & speedup  \\ 
			\hline
			enron & 3M & 2.1M & 30\% & 1.6M & 25\% & 656K & 59\% & 573 & 274 & 2.1x & 176 & 1.6x & 28 & 6.3x \\
			\hline
			gowalla & 3.2M & 2M & 38\% & 1.3M & 33\% & 848K & 39\% & 353 &  172 & 2.1x & 88 & 2.0x & 69 & 1.3x \\
			\hline
			patent & 3.5M & 2.4M & 31\% & 1.8M & 25\% & 1.6M & 11\%  & 3K & 1.4K & 2.1x & 700 & 2.0x & 524 & 1.3x   \\
			\hline
			road & 3.4M & 2.2M & 35\% & 1.7M & 22\% & 1.6M & 5\% & 2.4K & 675 & 3.6x & 456 & 1.5x & 456 & 1.0x  \\
			\hline
			WatDiv & 40M & 30M & 25\% & 21M & 28\% & 13M & 39\% & 43K & 31K & 1.4x & 25K & 1.2x & 4.4K & 5.7x \\
			\hline
			DBpedia & 53M & 31M & 42\% & 24M & 21\% & 14M & 43\% & 85K & 48K & 1.8x & 36K & 1.3x  & 6K & 6.0x \\
			\hline
		\end{tabular}
		\begin{tablenotes}
			\item[1] Basic GSI implementation with traditional CSR structure, two-step scheme and naive set operation.
			\item[2,3,4] Add techniques to GSI- one by one: PCSR structure, Prealloc-Combine strategy and GPU-friendly set operation.
		\end{tablenotes}
		\label{tab:GSI}
	\end{threeparttable}
	\vspace{-0.15in}
\end{table*}

\subsection{Evaluating Join Phase}

We evaluate three techniques of the join phase in GSI: PCSR structure, the Prealloc-Combine strategy and GPU-friendly set operation. 
Table \ref{tab:GSI} shows the result, where GSI- is the basic implementation with traditional CSR structure, two-step output scheme and naive set operation.
Two metrics are compared: (1) the number of transactions for reading data from global memory (GLD); (2) the time cost of answering subgraph search query.
We add techniques to GSI- one by one, and compare the performance of each technique with previous implementation.
For example, in Table \ref{tab:GSI}, the column ``+SO'' is compared with the column ``+PC'' to compute GLD drop and speedup.
After adding these techniques, we denote the implementation as GSI.

\begin{table}[htbp]
	\caption{\myhl{Comparison of CR and PCSR}}
	\vspace{-0.1in}
	\begin{threeparttable}
		\small
		\centering
		\begin{tabular}{|c|r|r|r|r|r|r|}
			\hline
			\multirow{2}*{Dataset} & \multicolumn{3}{|c|}{GLD} & \multicolumn{3}{|c|}{Time (ms)} \\ 
			\cline{2-7}
			~ & CR & PCSR & drop & CR & PCSR & speedup  \\ 
			\hline
			enron & 2.4M & 2.1M & 13\%  & 311 & 274 & 1.1x \\
			\hline
			gowalla & 2.5M & 2M & 20\% & 212 & 172 & 1.2x  \\
			\hline
			patent & 3.2M & 2.4M & 25\% & 1.8K & 1.4K  & 1.3x \\
			\hline
			road & 3.0M & 2.2M & 27\% & 873 & 675 & 1.3x \\
			\hline
			WatDiv & 46M & 30M &  35\% & 42K  & 31K &  1.4x \\
			\hline
			DBpedia & 37M & 31M & 16\%  & 56K & 48K &  1.2x \\
			\hline
		\end{tabular}
		\label{tab:cr}
	\end{threeparttable}
\end{table}

\subsubsection{Performance of PCSR structure}

To verify the efficiency of PCSR in Section \ref{sec:graphds}, we compare it with traditional CSR structure.
\nop{
As for ``Basic Representation'' (BR) and ``Compressed Representation'' (CR), we briefly report their performance: \\
(1) BR consumes too much memory and is unable to run on large graphs with hundreds of edge labels, thus we do not include it in the comparison.  \\
(2)  CR is memory-friendly but too slow both theorectically and experimentally, incuring several times larger GLD and longer time than PCSR. 
}
We set the bucket size as 128B and find that the maximum length of conflict list is below 15, even on the largest dataset.
Therefore, with PCSR structure, GSI always finds the address of $N(v,l)$ within one memory transaction, which is a big improvement compared to traditional CSR.

Table \ref{tab:GSI} shows that PCSR brings an observable drop of GLD (about 30\%),  and nearly 2.0x speedup.
The least improvement is observed on WatDiv due to small $|L_{E}|$, while on other datasets the power of PCSR is tremendous, achieving more than 1.8x speedup.
The superiority of PCSR is two-fold:
(1)  fewer memory transactions are needed, as presented in Table \ref{tab:csrs};
and (2)  threads are fully utilized while traditional CSR suffers heavily from thread underutilization.

\myhl{
We also compare PCSR with ``Compressed Representation'' (CR) in Table \ref{tab:cr}.
On all datasets, at least 13\% drop of GLD and 1.1x speedup are achieved.
WatDiv delivers the best performance, where CR even has higher GLD than traditional CSR.
PCSR has no advantage over CR when enumerating $N(v,l)$, and the only difference is locating.
WatDiv has the minimum number of edge labels, thus its edge label-partitioned graphs are very large, leading to high cost of locating for CR.
The ``Basic Representation'' (BR) consumes too much memory to run on large graphs with hundreds of edge labels.
}

\nop{
\begin{table}[htbp]
	\begin{threeparttable}
	\centering
	\caption{Performance of PCSR data structure}
	\begin{tabular}{|c|r|r|r|r|r|r|}
		\hline
		\multirow{2}*{Dataset} & \multicolumn{3}{|c|}{GLD\tnote{1}} & \multicolumn{3}{|c|}{Time (ms)} \\ 
		\cline{2-7}
		~ & csr\tnote{2} & pcsr & drop &  csr & pcsr & drop  \\ 
		\hline
		enron & 2M & 1,8M & 10\%  & 382 & 187 &  51\%  \\
		\hline
		gowalla & 2M & 1,8M & 12\% & 192 & 93 & 52\%   \\
		\hline
		road & 2.3M & 1.7M & 24\% & 1.6K & 523 & 68\%   \\
		\hline
		watdiv & 27M & 25M & 8\% & 29K & 27K & 6\%  \\
		\hline
		dbpedia & 29M & 26M & 10\% & 47K & 38K & 19\%  \\
		\hline
	\end{tabular}
\begin{tablenotes}
	\item[1] Number of Global Memory Load Transactions.
	\item[2] Gunrock CSR structure.
\end{tablenotes}
	\label{tab:pcsrPerf}
	\end{threeparttable}
\end{table}
}

\optionshow{}{
	\begin{table*}[htbp]
		\centering
		\caption{Performance of Write Cache}
		\begin{tabular}{|c|r|r|r|r|r|r|}
			\hline
			\multirow{2}*{Dataset} & \multicolumn{3}{|c|}{Global Memory Store Transactions} & \multicolumn{3}{|c|}{Query Response Time (ms)} \\ 
			\cline{2-7}
			~ & no cache & write cache & drop & no cache & write cache & drop  \\ 
			\hline
			enron & 25,371 & 23,056 & 9\% &  117 & 28 & 76\% \\
			\hline
			gowalla & 43,304 & 37,147 & 14\% & 78 & 69 & 12\%  \\
			\hline
			road & 70,430 & 65,500 & 7\% & 456 & 456 & 0\%  \\
			\hline
			WatDiv & 110,744 & 86,934 & 22\% & 8,396 & 4,425 & 47\%  \\
			\hline
			DBpedia & 248,670 & 90,284 & 64\% & 12,194 & 6,148 & 50\%  \\
			\hline
		\end{tabular}
		\label{tab:wrtcPerf}
	\end{table*}
}

\subsubsection{Performance of Parallel Join Algorithm}
In our vertex-oriented join strategy, there are two main parts: the Prealloc-Combine strategy (PC) and GPU-friendly set operation (SO). 

To evaluate Prealloc-Combine strategy, we implement the two-step  output scheme \cite{GpSM} as the baseline. 
Table \ref{tab:GSI} shows that on all datasets, PC obtains more than 21\% drop of GLD and 1.2x speedup. 
The gain originates from the elimination of double work during join, which also helps reduce GLD, thus further boosts the performance. 
It must be pointed out that PC can reduce the amount of work by at most half, thus there is no speedup larger than 2.0x.






\nop{
In other experiments, GSI is implemented using Prealloc-Combine strategy.
Only in this section we implement GSI by the two-step output scheme \cite{GpSM} and compare the performance of these two strategies.
The result is given in Table \ref{tab:joinPerf}, where we compare both the gld number and query response time.
On all datasets, at least 23\% drop can be found in gload memory transactions.
On the other hand, Prealloc-Combine claims a 30\% drop of time on enron, gowalla and road while on the other two datasets the drop is more than 20\%.

Experimentally and theorectically, Prealloc-Combine strategy not only helps reduce work complexity, but also lowers the gld number.
These prove the advantage of Prealloc-Combine compared with traditional two-step output scheme and lead to an unneglectable gain in the final performance.
}

To evaluate our GPU-friendly set operation, we compare with naive solution: finish each set operation with a new kernel function. 
Table \ref{tab:GSI} shows that SO reduces GLD by about 40\%; consequently, it leads to more than 1.3x speed up. 
\myhl{On patent and road, the improvement is not apparent because their neighbor lists are relatively small}.
SO also eliminates the cost of launching many kernel functions.

SO performs best on enron, WatDiv and DBpedia, showing >5.7x speedup.
The reason is that write cache performs best on these graphs, thus saving lots of global memory store transactions (GST). 
On other graphs, the gain of write cache is small because they have fewer matches, thus perform fewer write operations.
\nop{
Another key reason for performance improvement is that we use write cache to save lots of global memory store transactions (GST). 
Table \ref{tab:GSI} shows that on some datasets (enron, WatDiv and DBpedia), the speedup is more than 5.7x, which is due to the reduced GST by write cache.
The speedup on gowalla and road is small because they have fewer matches, thus perform fewer writing operations and yield a small GST even without write cache.
}
\optionshow{
More details  are in our full paper \cite{fullVersion}.
}{}

\nop{
gives the comparison result of two kinds of set operations.
Clearly, the performance gain is extremely large when we utilize \emph{Shared Memory} with batch implementation.
Especially, on watdiv100M and dbpedia170M, the gld number is reduced by nearly 40\% and the reduction of query time is more than 60\%.
On watdiv100M, the query response time is even lowered by a magnitude.
The batch implementation is not very complicated but its efficiency demonstrates the significance of utilizing \emph{Shared Memory}.

We compare our batch implementation of set operations in Algorithm \ref{alg:joinTwo} with basic implementaion.
In the basic implementation, we do not use batches but follow the general ideas, i.e. read elements one by one.
In Line \ref{algcmd:subtract}, we read an element $x$ each time and use a loop to check if $x$ is already included in the partial answer.
In Line \ref{algcmd:intersect}, we do binary search on candidate set $C$ to check if some element exists.
We use binary search instead of merge-intersection here because $C$ is usually much larger then $buf_{i}$.
However, in Line \ref{algcmd:merge} we use merge-intersection because the two lists are both sorted and they have the same scale.

Table \ref{tab:setopPerf} gives the comparison result of two kinds of set operations.
Clearly, the performance gain is extremely large when we utilize \emph{Shared Memory} with batch implementation.
Especially, on watdiv100M and dbpedia170M, the gld number is reduced by nearly 40\% and the reduction of query time is more than 60\%.
On watdiv100M, the query response time is even lowered by a magnitude.
The batch implementation is not very complicated but its efficiency demonstrates the significance of utilizing \emph{Shared Memory}.
In this section we check the efficiency of write cache(proposed in Section \ref{sec:wrtc}) by comparing global memory store transactions(gst for short) and query response time.
The baseline is without write cache, which means GSI writes intermediate results to global memory directly each time no matter how large they are.
The result of experiments is placed in Table \ref{tab:wrtcPerf}.

Apparently, the write cache brings performance gains by reducing gst number.
On gowalla, the gst number is reduced by 14\% while on watdiv100M and dbpedia170M the number is lowered by a magnitude.
On watdiv100M, query time is reduced by nearly a half.
On dbpedia170M, time cost is even lowered by a magnitude.
These improvements proves the effectiveness of write cache, which is also the kind of optimization by utilizing \emph{Shared Memory}.
}

\subsection{Evaluating Optimization Techniques in GSI}
We evaluate the two optimization strategies proposed in Section \ref{sec:optimize} and give the result in Table \ref{tab:GSI-opt}, where column ``+DR'' is compared with column ``+LB''.
After adding the two optimizations, we denote the implementation as GSI-opt.
  
\subsubsection{Performance of Load Balance scheme}

The 4-layer balance scheme (LB) in Section \ref{sec:balance} does not save global memory transactions, or the amount of work.
However, it improves the performance by assigning workloads to GPU processors in a more balanced way.
We verify its efficiency by comparing it with the strategy used in \cite{DBLP:conf/ppopp/MerrillGG12}.
\optionshow{
The study of tuning parameters is in \cite{fullVersion}.
}{}
$W_{2}$ should be equal to the block size of CUDA (1024), and empirically we set $W_{1}=4W_{2}=16W_{3}=4096$.

On the four smaller datasets, LB does not show much advantage because the time cost is already very low (less than 0.6 seconds) and the load imbalance is slight.
But on other datasets, LB brings tremendous performance gain, i.e., more than 2.7x speedup.
This demonstrates that our strategy is especially useful on large scale-free graphs, due to the existence of severely skewed workloads.
\myhl{Note that though patent is scale-free, its maximum degree is small, which limits the effect of LB.}

\begin{table}[htbp]
		\caption{Performance of optimizations}
				\vspace{-0.1in}
	\begin{threeparttable}
	\small
	\centering
	\begin{tabular}{|c|r|r|r|r|r|}
		\hline
		Dataset$\backslash$Time(ms) & GSI & +LB\tnote{1} & speedup & +DR\tnote{2} & speedup \\
		\hline
		enron &  28 & 28 &  1.0x & 28 & 1.0x \\
		\hline
		gowalla & 69 & 69 & 1.0x  & 68 & 1.0x \\
		\hline
		patent & 524 & 466 & 1.1x & 465 & 1.0x    \\
		\hline 
		road & 456 & 456 & 1.0x & 456 & 1.0x  \\
		\hline
		WatDiv & 4.4K & 1.3K & 3.4x  & 1K & 1.3x \\
		\hline
		DBpedia & 6K & 2.2K & 2.7x  & 2K & 1.1x \\
		\hline
	\end{tabular}
\begin{tablenotes}
	\item[1] Add load balance techniques to GSI.
	\item[2] Add duplicate removal method to GSI + LB.
\end{tablenotes}
	\label{tab:GSI-opt}
\end{threeparttable}
\end{table}

\subsubsection{Performance of Duplicate Removal method}

Using the duplicate removal method (DR) in Section \ref{sec:dupr}, input is shared within a block so the amount of work should be reduced theorectically.
Compared with baseline (no duplicate removal), Table \ref{tab:GSI-opt} shows  1.3x and 1.1x speedup on WatDiv and DBpedia, respectively.
Besides, GLD is also lower with DR, but its comparison is omitted here.

This experiment shows that DR really works, though the improvement is small.
The bottleneck is the region size that DR works on, i.e., a block.
Even with the block size set to maximum (1024), DR can only remove duplicates within 32 rows since we use a warp for each row.

\optionshow{}{
\begin{table*}[htbp]
	\centering
	\caption{Performance of Duplicate Removal method}
	\begin{tabular}{|c|r|r|r|r|r|r|}
		\hline
		\multirow{2}*{Dataset} & \multicolumn{3}{|c|}{Global Memory Load Transactions} & \multicolumn{3}{|c|}{Query Response Time(ms)} \\ 
		\cline{2-7}
		~ & with duplicates & duplicate removal & drop & with duplicates & duplicate removal & drop  \\ 
		\hline
		enron & 656K & 630K & 4\% & 28 & 28 & 0\%  \\
		\hline
		gowalla & 848K & 824K & 3\% &  69 & 68 & 1\% \\
		\hline
		road & 1.66M & 1.61M & 3\% &  456 & 456 & 0\%  \\
		\hline
		WatDiv & 13M & 10M & 21\% &  1.3K & 1K &  17\%  \\
		\hline
		DBpedia & 14M & 10M & 23\% &  2.2K & 2K &  9\%  \\
		\hline
	\end{tabular}
\vspace{-0.15in}
	\label{tab:duprPerf}
\end{table*}
}

\subsection{Comparison of GSI with counterparts} \label{sec:counterpart}

\Paragraph{Overall Performance}.
The results are given in Figure \ref{fig:allTime}, where GSI and GSI-opt represent implementations without and with optimizations (in Section \ref{sec:optimize}), respectively.
Note that there is no bar  if the corresponding time exceeds the threshold of 100 seconds.
In all experiments, GPU solutions beat CPU solutions as expected due to the power of massive parallelism.

\myhl{Considering existing GPU solutions only, there is no clear winner between four counterparts, but they all fail to compete with GSI.}
GSI runs very fast on the first four datasets, answering queries within one second.
On WatDiv and DBpedia, GSI achieves more than 4x speedup over counterparts. 

Focusing on our solution, on the first four datasets, GSI-opt is close to GSI; while on the latter two, GSI-opt shows more than 3x speedup.
To sum up, our solution outperforms all counterparts on all datasets by several orders of magnitude.

%
%
%

\begin{figure}[htbp]   
	\centering
	\includegraphics[width=8cm]{\picfolder 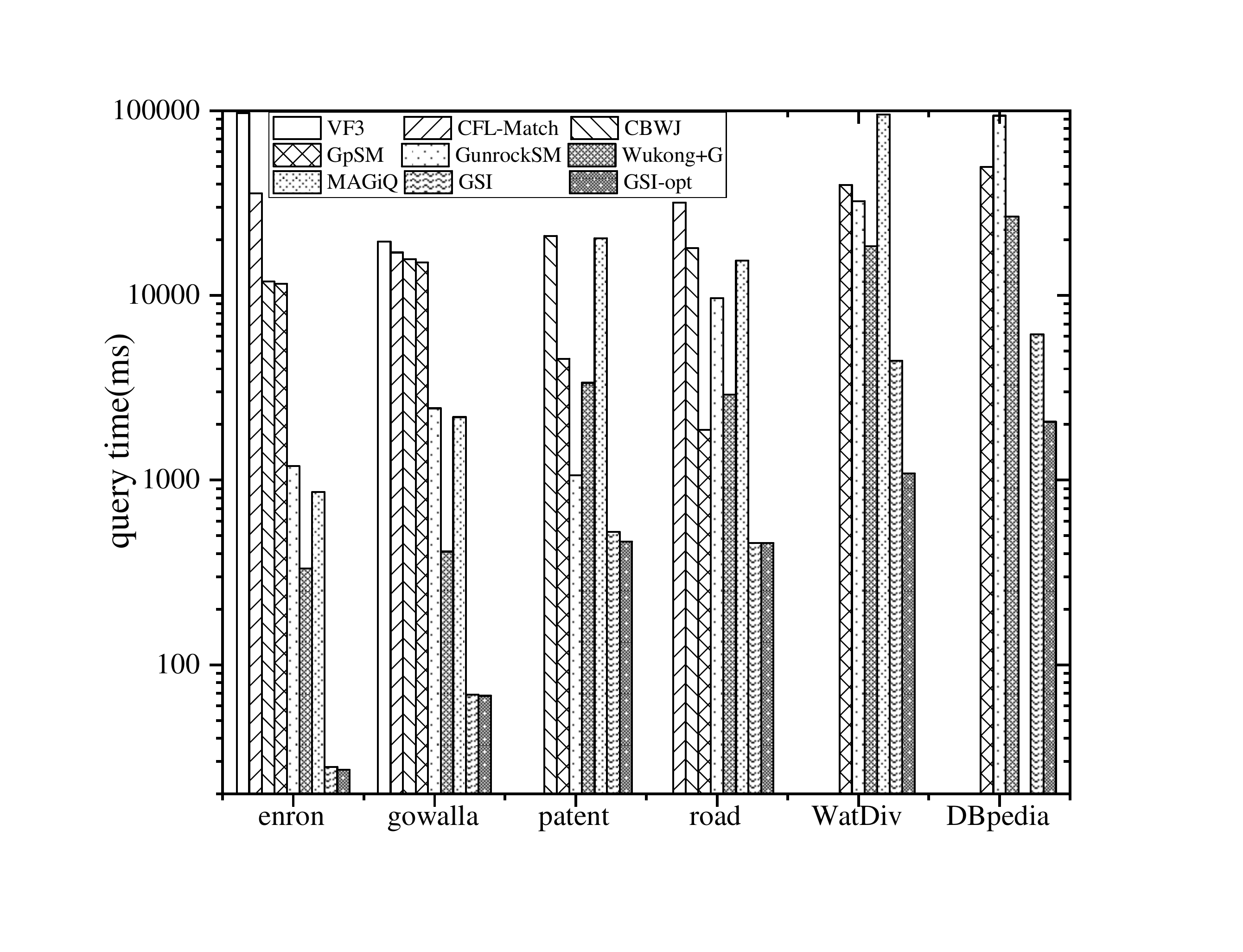}   
		\vspace{-0.1in}  
			\caption{\myhl{Performance Comparison on all datasets}}      	
			 	\label{fig:allTime}     	 
\end{figure}

\Paragraph{Scalability}.
We generate a series of RDF datasets using the WatDiv benchmark.
These scale-free graphs are named watdiv10M, watdiv20M,...,watdiv210M, with the number of vertices and edges growing linearly as the number in the name. CPU solutions fail to run even on the smallest watdiv10M dataset, thus, we only compare GPU solutions and draw the curves in Figure \ref{fig:scala}. 
\myhl{
Note that the curve stops if the memory consumption exceeds  GPU capacity or the time exceeds 100s.
The curves of four counterparts are  above the curves of GSI and GSI-opt.
Besides, they rise sharply as the data size grows larger.
In contrast, GSI-opt rises much more slowly. 
}

\nop{From Figure \ref{fig:scala}, we know that our solution (GSI-opt) can scale on a large graph with 210 million edges on WatDiv benchmark, which is much larger than the capability of existing solutions.}

\myhl{
After watdiv210M, all algorithms fail on most queries due to the limitation of GPU memory capacity.
The counterparts stop much earlier because they have larger candidate tables and intermediate tables.
Due to its efficient filter, GSI occupies less memory, thus scaling to larger graphs.
Furthermore, during each join iteration of GSI,  only an edge label-partitioned graph is needed on GPU.
In summary, GSI not only outperforms others by a significant margin, but also shows  good scalability so that graphs with hundreds of millions of edges for subgraph query problem are now tractable.
}

\begin{figure}[htbp]
	\centering	
	\subfigure[{\tiny \myhl{Scalability with the data size}}]
	{
		\label{fig:scala}
		\begin{minipage}[c]{4cm}		
			\centering		
			\includegraphics[width=4cm]{\picfolder 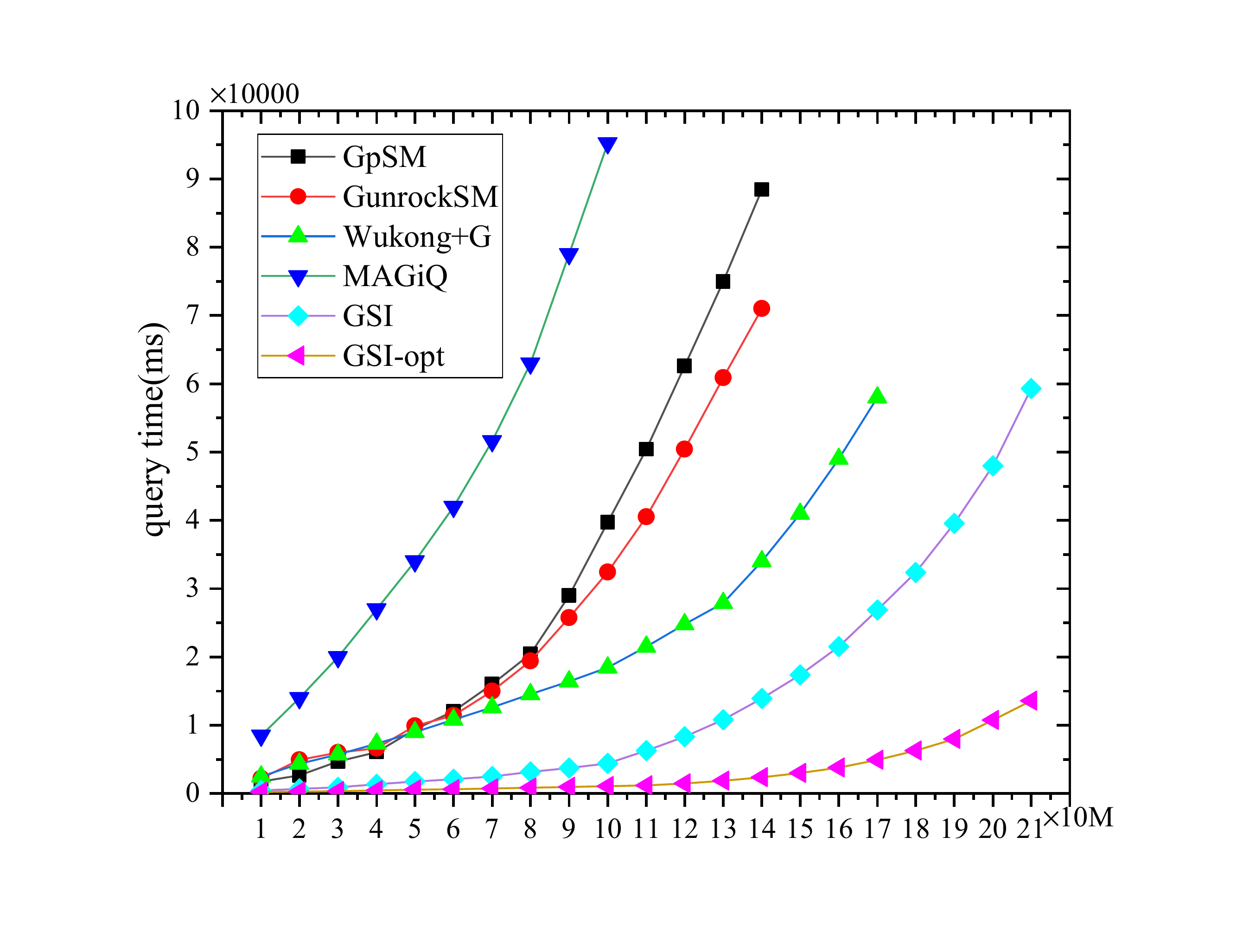}		
		\end{minipage}	
	}
	\subfigure[{\tiny \myhl{Scalability with the number of SMs}}]
{
	\label{fig:scalaSM}   
	\begin{minipage}[c]{4cm}		
		\centering		
		\includegraphics[width=4.2cm]{\picfolder 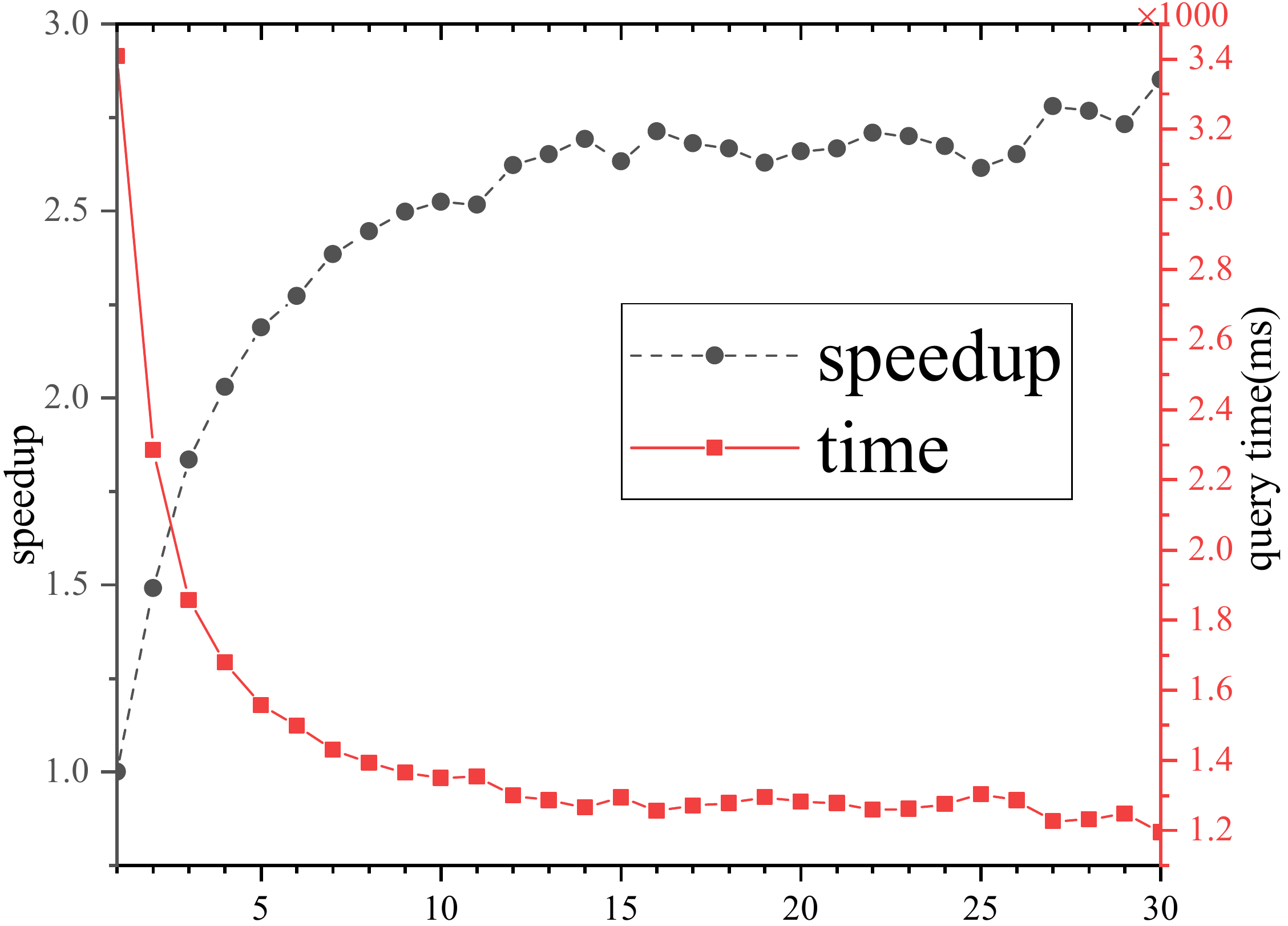}		
	\end{minipage}	
}
	\vspace{-0.1in}
	\caption{\myhl{Scalability Test}	}
\end{figure}

\myhl{
We also evaluate the scalability of GSI with the number of SMs.
In order to control the number of running SMs, we limit the number of blocks launched and withdraw dynamic kernels (see Section \ref{sec:balance}), which may degrade the performance.
We choose WatDiv (see Table \ref{tab:dataset}) and show the result in Figure \ref{fig:scalaSM}.
With the number of SMs increasing from 1 to 30, the response time drops continuously, though with some tiny fluctuations.
The time curve drops fast in the beginning, but slows down gradually, corresponding to the sub-linear speedup curve.
The maximum speedup is 2.85, and is limited by the irregularity of graphs and GPU memory bandwidth, which cause severe load imbalance and high memory latency.}

\subsection{Vary the number of labels and query size} \label{sec:lqsize}

In this section, we explore the influence of number of labels and query size.
We use GSI-opt and select gowalla as the benchmark.
By default, the number of vertex and edge labels are both 100, and all queries have 12 vertices.

We vary the number of labels and show results in Figure \ref{fig:label}.
As the number of labels increases, run time decreases.
The ``vertex label num'' line shows sharper drop because larger $|L_{V}|$ directly reduces the sizes of candidate sets.
However, after $|L_{V}|>100$, the drop quickly slows down to zero as candidate sets are small enough to be fully parallelized.
Similarly, larger $|L_{E}|$ also helps reduce $|C(u)|$ due to improved pruning power of labeled edges.
In addition, the size of $|N(v,l)|$ is also lowered as $|L_{E}|$ grows.
This is the reason that run time keeps dropping, though the speed also changes after $|L_{E}|>100$.

\begin{figure}[htbp]
	\centering
	\subfigure[{\tiny Vary the number of vertex and edge labels}]
	{
		\label{fig:label}   
		\begin{minipage}[c]{4cm}		
			\centering		
			\includegraphics[width=4cm]{\picfolder 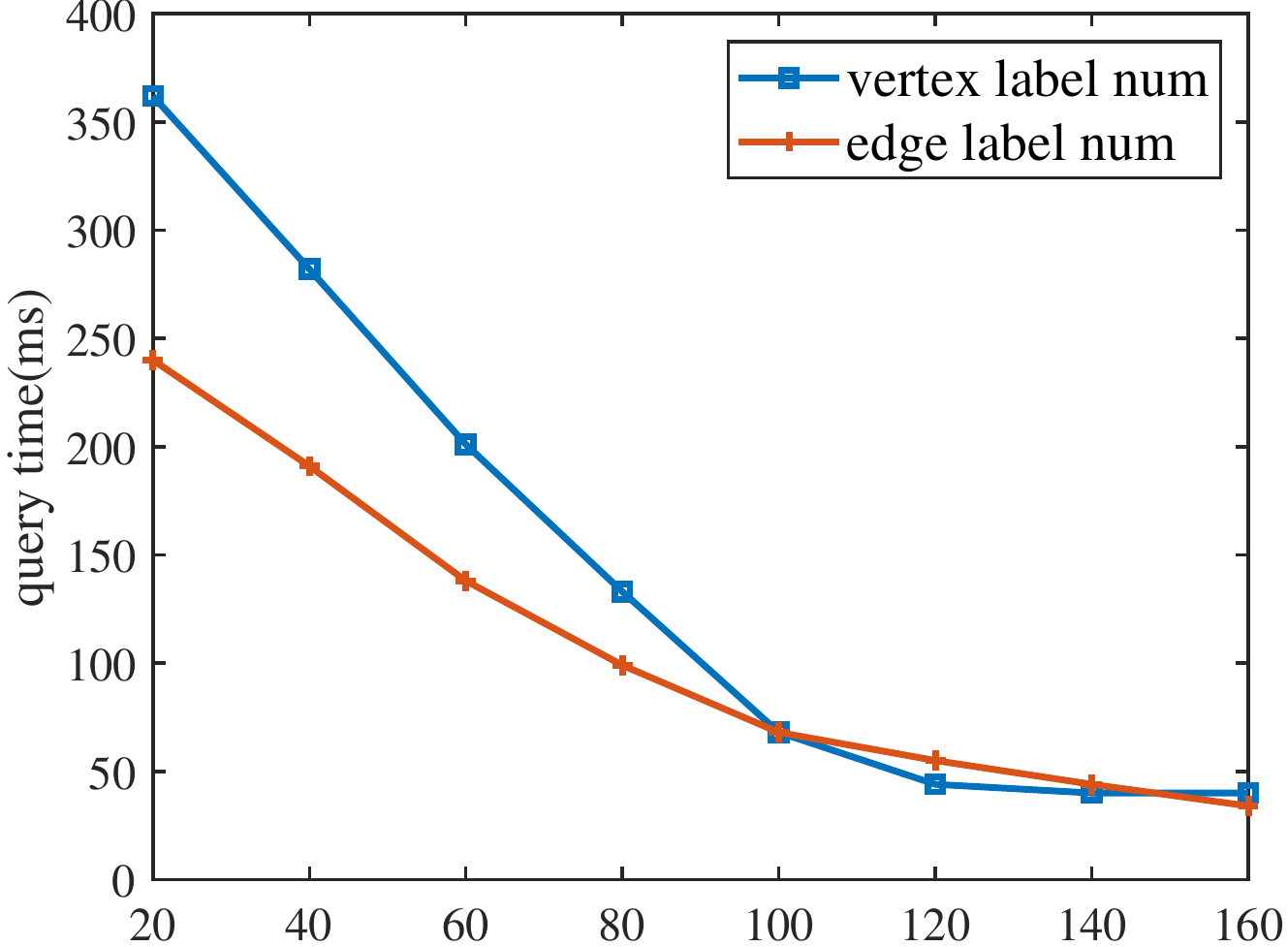}		
		\end{minipage}	
	}
	\subfigure[{\tiny Vary the number of edges and vertices in $Q$}]
	{
		\label{fig:qsize}
		\begin{minipage}[c]{4cm}		
			\centering		
			\includegraphics[width=4cm]{\picfolder 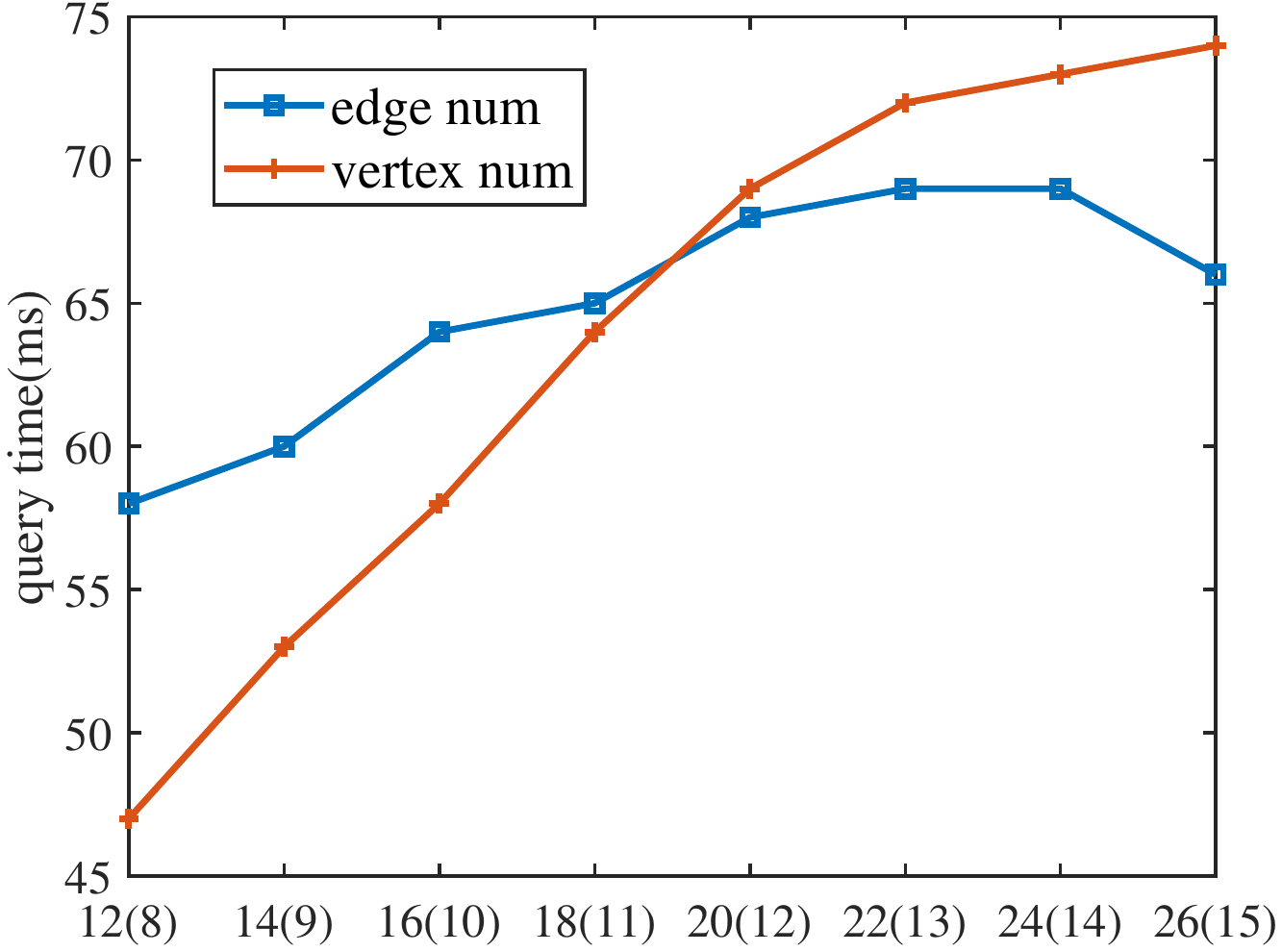}		
		\end{minipage}	
	}
	\vspace{-0.1in}
	\caption{Experiments of label number and query size}	
\end{figure}

As for query size, we first fix $|V(Q)|=12$ and vary the number of edges, then fix $|E(Q)|=2\times |V(Q)|$ and vary the number of vertices.
Figure \ref{fig:qsize} shows the result, where unenclosed X-axis numbers denote the number of edges, and the X-axis numbers enclosed in parentheses denote the number of vertices.
In the first case, run time rises slowly, because the cost of processing extra edges is marginal.
After $|E(Q)|>24$, a small drop occurs as there are enough edges to provide stronger pruning potential.
In the second case, an observable increase can be found because in our vertex-oriented join strategy, larger $|V(Q)|$ means more join iterations.
However, the rise slows down after $|V(Q)|\ge 13$.
Generally, larger query graph results in fewer matches, thus the cost of each join iteration is lower.

\subsection{Distribution of query time and result size} \label{sec:distribution}
\myhl{
In addition to average time, using GSI-opt, we also show the distribution of query  time and result size in Figure \ref{fig:distribution}.
}

\myhl{
For query time, the relative height of boxes corresponds to the results of average time in Figure \ref{fig:allTime}.
Due to the irregularity of graphs, a few outliers exist.
On WatDiv and DBpedia, the mean is above the major part (the box area) because the cost of outliers is too high.
}

\myhl{
For result size, gowalla, patent and road deliver the minimum value, which limits the speedup by write cache.
It must be noted that query time is not decided by result size.
Besides, the data skew on the latter two graphs are more severe, thus outliers are far above the box.
}

\begin{figure}[htbp]
	\centering
	\subfigure[{\tiny \myhl{Box plot of query response time} }]
	{
		\label{fig:timeBox}   
		\begin{minipage}[c]{4cm}		
			\centering		
			\includegraphics[width=4cm]{\picfolder 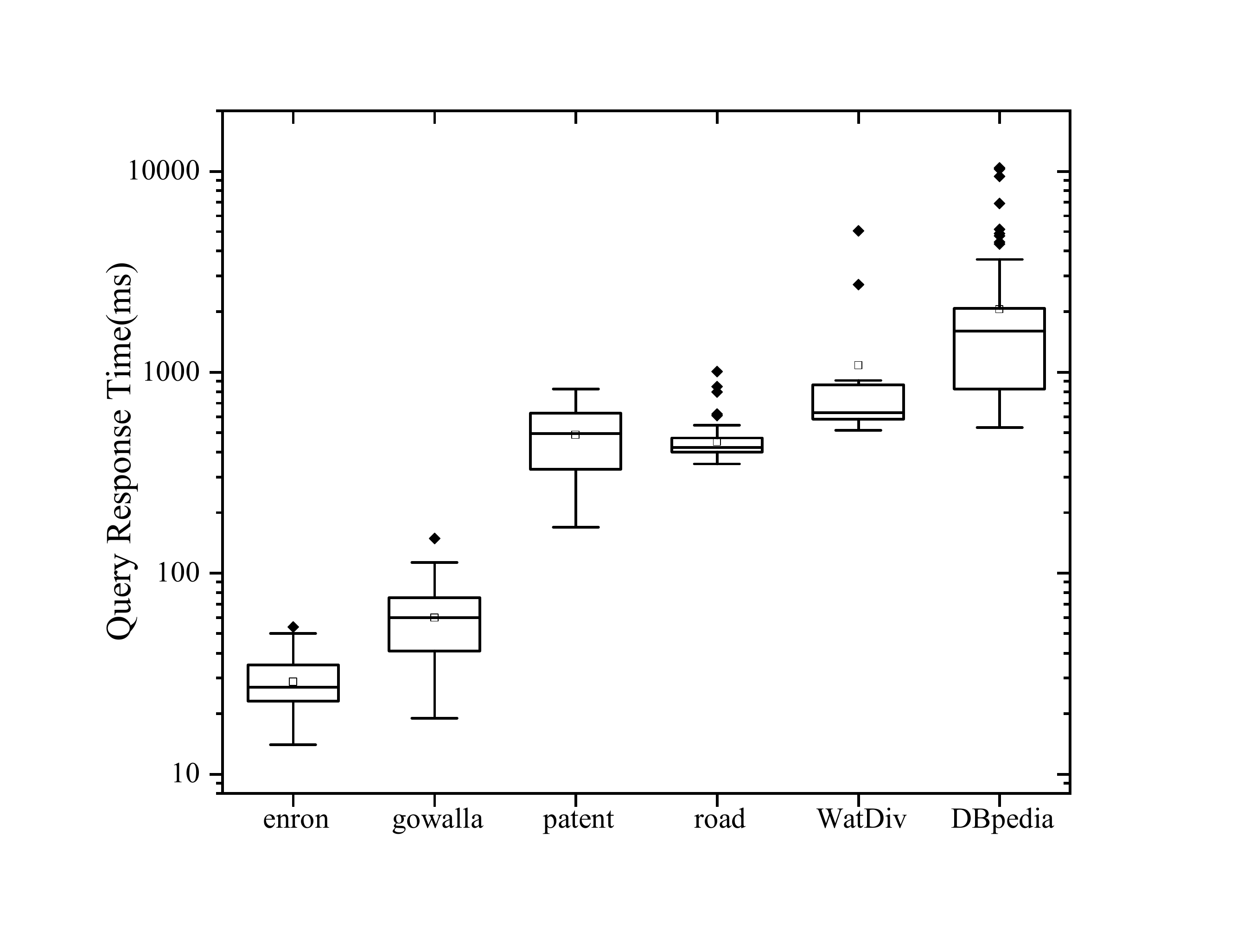}		
		\end{minipage}	
	}
	\subfigure[{\tiny \myhl{Box plot of query result size} } ]
	{
		\label{fig:sizeBox}
		\begin{minipage}[c]{4cm}		
			\centering		
			\includegraphics[width=4cm]{\picfolder 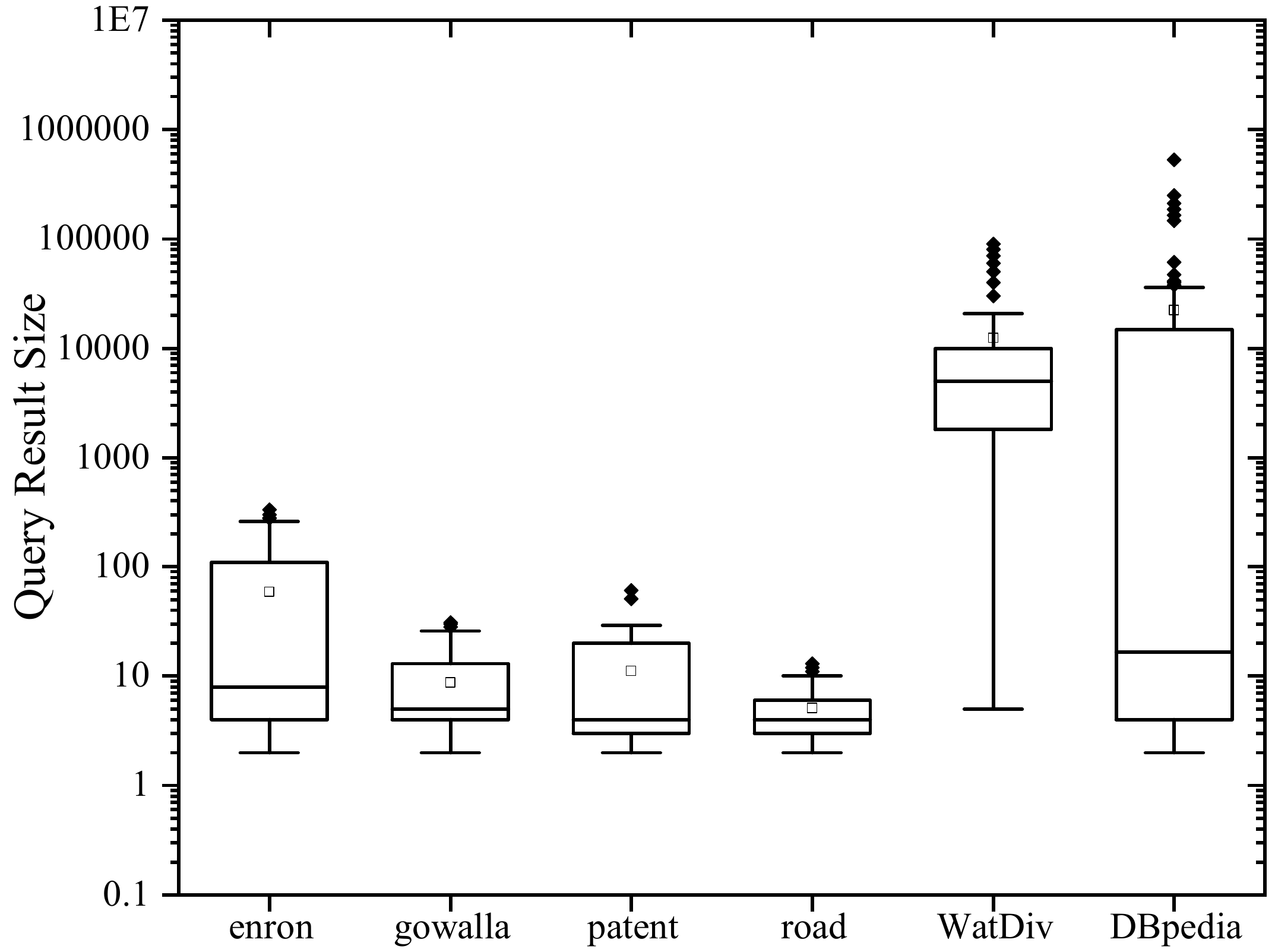}		
		\end{minipage}	
	}
	\vspace{-0.1in}
	\caption{\myhl{Distribution of query response time and query result size} }	
	\label{fig:distribution}
\end{figure}

\subsection{Additional Experiments} \label{sec:addexp}

\myhl{
We report the loading time (from host to GPU) of GSI on all datasets: 1ms, 5ms, 106ms, 120ms, 144ms, 178ms.
Besides, we record the maximum memory consumption of GPU algorithms in Table \ref{tab:gpumem}, including host and GPU memories.  
For CPU algorithms, only host memory consumption is reported in Table XI.
Note that ``NAN'' means an algorithm cannot end in a reasonable time.
Besides, CPU solutions all fail on WatDiv and DBpedia.
Obviously, backtracking solutions (VF3 and CFL-Match) occupy less memory.
}

\myhl{
GPU solutions that are based on BFS have larger memory consumption on both host and GPU.
Compared to counterparts,  GSI consumes more host memory due to the maintainance of signature table and PCSR structures.
However, GSI consumes less  GPU memory because it has smaller candidate/intermediate tables and in each iteration only an edge label-partitioned graph is needed on GPU.
}

\begin{table*}[htbp]
	\caption{\myhl{Memory consumption of GPU algorithms}}
	\vspace{-0.1in}
	\begin{threeparttable}
		\small
		\centering
		\begin{tabular}{|c|r|r|r|r|r|r|r|r|r|r|}
			\hline
			\multirow{2}*{Dataset} & \multicolumn{5}{|c|}{Host Memory} & \multicolumn{5}{|c|}{GPU Memory} \\ 
			\cline{2-11}
			~ & GpSM & GunrockSM & Wukong+G & MAGiQ & GSI & GpSM & GunrockSM & Wukong+G & MAGiQ & GSI   \\ 
			\hline
			enron & 154M &  181M & 174M & 284M & 160M   &  1.3G & 1.4G & 667M & 721M & 661M \\
			\hline
			gowalla & 592M & 712M & 599M & 367M &  466M  &  2.4G & 2.8G & 1.8G & 725M &  685M  \\
			\hline
			patent & 1.0G & 1.3G & 1.3G & 1.5G &  663M  & 3.5G  & 3.7G  & 2.1G & 1.1G &   915M  \\
			\hline
			road & 1.8G & 2.0G & 2.1G & 2.2G &  1.1G  & 3.6G  & 3.6G & 2.3G & 1.3G &  1.2G \\
			\hline
			WatDiv & 4.4G & 5.7G & 4.7G & 6.9G &  8.5G   &  7.3G & 7.7G & 7.5G & 5.6G &  4.9G  \\
			\hline
			DBpedia & 6.9G & 8.2G & 8G & NAN &  14G   &  9.0G & 9.6G & 9.6G & NAN &  8.1G  \\
			\hline
		\end{tabular}
		\label{tab:gpumem}
	\end{threeparttable}
	\vspace{-0.25in}
\end{table*}

\nop{
When answering a query $Q$ on $G$, we need to transfer  the signature table (for filtering) and the PCSR structures (for joining) to GPU memory.
The process above needs to be performed for each query, and its time cost is the loading time, which is included in the measurement.
The loading time  is small (<0.2s), and the details are given in Table \ref{tab:loadtime}.
}

\begin{table}[htbp]
	\begin{threeparttable}
		\caption{\myhl{Memory consumption of CPU algorithms}}
		\vspace{-0.15in}
		\begin{tabular}{|c|r|r|r|r|}
			\hline
			Algorithm$\backslash$Dataset & enron & gowalla & patent & road   \\
			\hline
			VF3  & 10M & 39M &  NAN & NAN      \\
			\hline
			CFL-Match & 21M & 62M  &  NAN & 416M   \\
			\hline
			CBWJ & 15M & 64M  &  817M & 1.5G   \\
			\hline
		\end{tabular}
	\end{threeparttable}
	\label{tab:cpumem}
\end{table}

\nop{
Data graph is loaded into memory first and preprocessed to build encoding table and PCSR structures, after which the queries are read one by one and executed.
For each query, we begin counting the time after the query graph is read into memory and end when all results are flushed to a disk file.
We report error if some algorithm fails to answer a query within 100 seconds.
Time of building encoding table and PCSR structures is not listed due to limit of pages.
On the other hand, preprocessing is done offline and consumes several hundreds of seconds at most in our experiments.
}

\nop{
\subsection{Comparison with CBWJ} \label{sec:cbwj}

\myhl{CBWJ \cite{CBWJ}}

the state-of-the-art CPU algorithm of  WOJ, i.e., CBWJ [32]. 
However, the code given by the author of CBWJ fails to run our datasets due to large number of edge labels.
Thus, we re-generate datasets with small number of edge labels to finish this comparison.
Results are given in full version [19].
Thus, for  subgraph query, we compare the non-relational approaches (CFL-Match and VF3) and the relational approach (CBWJ) in Figure ??. 
}

%% file: related.tex
\section{Related Work}\label{sec:related}

\Paragraph{CPU-based subgraph isomorphism}.
Ullmann \cite{DBLP:journals/jacm/Ullmann76} and VF2 \cite{DBLP:journals/pami/CordellaFSV04} are the two early efforts; Ullmann uses depth-first search strategy, while VF2 considers the connectivity as pruning strategy.
Most later methods (e.g., \cite{DBLP:journals/pvldb/ZhaoH10,  DBLP:journals/pvldb/ShangZLY08}) pre-compute some structural indices to reduce the search space and optimize the matching order using various heuristic methods.
TurboISO \cite{DBLP:conf/sigmod/HanLL13} merges similar query nodes and BoostISO \cite{DBLP:journals/pvldb/RenW15} extends this idea to data graph.
CFL-Match \cite{CFL-Match} defines a Core-Forest-Leaf decomposition  and select the matching order based on minimal growth of intermediate table.
VF3 \cite{VF3} is an improvement of VF2, which leverages more pruning rules (node classification, matching order, etc.) and favors dense queries. \myhl{
EmptyHeaded \cite{EmptyHeaded} and CBWJ \cite{CBWJ} are based on worst-case optimal join \cite{WOJ}; CBWJ achieves better performance.
}
Unfortunately, these sequential solutions perform terrible on large graphs, due to exponential search space.

\Paragraph{GPU-based subgraph isomorphism}.
The first work  is  \cite{DBLP:conf/adc/LinZWWQ14}, which finds candidates for STwigs \cite{trinity} first and joins them.
However, STwig-based framework  may not be suitable for GPU due to large intermediate results.
Later, GPUSI \cite{yang2015gpu} transplants TurboISO to GPU.
Different candidate regions are searched in parallel, but its performance is limited by depth-first search within each region.
All backtracking-based GPU algorithms have problems of warp divergence and uncoalesced memory access, as analyzed in \cite{DBLP:conf/europar/JenkinsAOCS11}.

GpSM \cite{GpSM} and GunrockSM \cite{GunrockSM}  outperform previous works by leveraging \emph{breadth-first search}, which favors parallelism.
Their routines are already introduced in Section \ref{sec:introduction}.
They both adopt two-step output scheme to write join results, and do not utilize features of GPU architecture.
Therefore, they have problems of high volume of work, long latency of memory access and severe workload imbalance.
In summary, GpSM and GunrockSM both lack optimizations for challenges presented in Section \ref{sec:challenge}.

\myhl{
MAGiQ \cite{MAGiQ} and Wukong+G \cite{Wukong+G} are two GPU-based RDF systems that supports SPARQL queries.
Wukong+G develops efficient swapping mechanism between CPU and GPU, while MAGiQ utilizes existing CUDA libraries of linear algebra for filtering.
They have no optimization for table join, which marks them inefficient. 
}

\nop{
Existing work realted to subgraph isomorphism can be mainly divided into two categories: CPU-based and GPU-based.
Algorithms on CPU face great challenge of efficiency and scalability when dealing with very large graphs. 
As a result, subgraph matching is a bottleneck for overall performance in many applications. 

In order to yield better performance, several approaches are explored in the literature.
One way is to develop more efficient pruning strategy and select better matching order \cite{DBLP:journals/jacm/Ullmann76,DBLP:journals/pami/CordellaFSV04,DBLP:journals/pvldb/ZhaoH10,DBLP:conf/dasfaa/ZhuZLZW10, DBLP:journals/pvldb/LeeHKL12,DBLP:conf/sigmod/HanLL13,DBLP:journals/pvldb/KimSHHC15,DBLP:journals/pvldb/RenW15,DBLP:journals/jsc/McKayP14,DBLP:journals/pvldb/ShangZLY08,DBLP:conf/sigmod/BiCLQZ16}.
This will bring some improvements, but for graphs with tens of millions of nodes, it is not good enough.
Another approach is to develop distributed systems \cite{gStoreD, trinity}, which eases the burden of computation of a single machine.
However, data placement is a big challenge in distributed systems, as the communication cost among different machines may be too high to carry on.
The final approach is using new hardwares with power of massive parallism like GPU.
Recently, GPUs have been successfully leveraged for fundamental graph operations(like BFS  \cite{DBLP:conf/ppopp/MerrillGG12, DBLP:conf/sc/BeamerAP12, DBLP:conf/sc/LiuH15, DBLP:conf/sigmod/LiuHH16, DBLP:conf/ipps/PanPO18} and SSSP \cite{gpuBellmanFord, gpuSSSP}) on large graphs.

\subsection{CPU-based subgraph isomorphism algorithms}

Ullmann \cite{DBLP:journals/jacm/Ullmann76} and VF2 \cite{DBLP:journals/pami/CordellaFSV04} are the two early efforts for subgraph isomorphism problem, and the key problem is how to prune invalid matches as early as possible.
Ullmann uses depth-first search strategy, while VF2 considers the connectivity as pruning strategy.
In order to speed up query response time, most subgraph search methods \cite{DBLP:journals/pvldb/ZhaoH10, dasfaa:nova, ssdbm:aep, DBLP:journals/pvldb/ShangZLY08} pre-compute some structural indices to reduce the search space and optimize the matching order using various heuristic methods.
They assume that the data graph is a vertex-labeled graph, and build indices based on vertex labels and neighborhood structure.

Jinsoo Lee et al. \cite{DBLP:journals/pvldb/LeeHKL12} re-implements some of the above methods and provides a fair comparison of algorithms above. 
Then they present TurboISO \cite{DBLP:conf/sigmod/HanLL13}, which merges similar query vertices and enumerates all paths to find the best matching order.
BoostISO \cite{DBLP:journals/pvldb/RenW15} extends the concept of neighborhood equivalence class in the data graph and defines four types of relationships between vertices in the data graph to further reduce duplicate computation.
CFL-Match \cite{CFL-Match} defines a Core-Forest-Leaf decomposition strategy and a cost model to select a good matching order based on minimal growth of result size.
VF3 \cite{VF3} is an improvement of VF2, which leverages more pruning rules (node classification, match order, etc.) and favors dense queries.
These CPU algorithms only do experiments on real graphs with thousands of nodes, which are too small in real world.
In addition, TurboISO and CFL-Match consume much more memory, which limits their scalability.

\subsection{GPU-based subgraph isomorphism algorithms}


The first work of subgraph isomorphism on GPU is done in \cite{DBLP:conf/adc/LinZWWQ14}, which finds candidates for STwigs \cite{trinity} first and joins these candidates to get the final result.
It uses hash method instead of binary search to join two relation tables of STwigs, indicating a constant time cost.
However, hash methods always bring heavy issue of random memory access which limits the performance.
In addition, framework based on STwigs may not fit for GPU due to large intermediate results.

Later, GPUSI \cite{yang2015gpu} modifies TurboISO to fit for GPU.
Search processes in different candidate regions are run in parallel.
However, the whole framework is still based on backtracking paradigm and the overall performance is limited by depth-first search.


\Paragraph{GpSM}.
GpSM \cite{GpSM} outperforms previous works by leveraging breadth-first search to explore search space, which favors parallelism.
It collects candidates for query vertices and query edges first, later uses a refining process to reduce the number of candidates by the neighborhood restrictions.
Next, candidate edges are joined one by one following a BSP model.
In each step, intermediate table is joined with the candidate set of a query edge.
GpSM uses a warp for each row and adopts \emph{two-step output scheme} \cite{DBLP:conf/IEEEpact/HeFLGW08} to concurrently write matched results to a new intermediate table.
Shortcomings of GpSM can be summarized in three aspects:
\begin{itemize}
\item Two-step output scheme is used when collecting edge candidates and joining intermediate results. 
With this scheme, the join process is done twice in order to find the writing address of each result.
Double work amount marks it inefficient.
\item Caches are not fully utilized. \emph{Shared Memory} is only used to cache the content of intermediate table in a trival way. 
Each block can deals with 32 rows at most and GpSM only uses a tiny part(<1KB) of shared memory.
The size of shared memory on modern GPU is usually 48KB or larger, but GpSM does not utilize them.
\item Though GpSM uses a warp for each row, load imbalance still exists in the candidate collecting phase and the joining phase. 
For example, the workloads of different blocks is not balanced.
\end{itemize}

\Paragraph{GunrockSM}.
GunrockSM \cite{GunrockSM} is based on \cite{gunrock} and handles the subgraph isomorphism problem as follows:
\begin{enumerate}
\item filter by label and degree to collect candidates for each query node 
\item find the candidate set for each query edge
\item enumerate all combinations of all query edges to acquire the final result
\end{enumerate}
In the last two steps, GunrockSM also uses \emph{two-step output scheme}.
Due to the problem of uncertain result number, GunrockSM uses a big array to store the statuses of all combinations.
Later it does all enumerations again and adds valid matches to result table according to the status array.
Brute-force enumeration is the bottleneck, limiting the algorithm in two aspects.
On the one hand, many invalid matches are considered and lots of threads are wasted.
On the other hand, the size of status array may exceed the hardware limit.
Besides, GunrockSM does not adapt to the features of GPU very well.
Load imbalance occurs in the second step due to varied lengths of adjacent lists.
Furthermore, \emph{Shared Memory} is not leveraged to reduce memory latency.

To sum up, GpSM and GunrockSM both lack optimizations for challenges presented in Section \ref{sec:challenge}.
GSI designs novel and efficient methodologies to handle these problem, thus outperforms the literature by several orders of magnitude.
}

%% file: conclusion.tex
\section{Conclusions}\label{sec:conclusion}

We introduce an efficient algorithm (GSI),  utilizing GPU parallelism for large-scale subgraph isomorphism.
GSI is based on filtering-and-joining framework and optimized for the architecture of modern GPUs.
Experiments show that GSI outperforms all counterparts by several orders of magnitude. 
Furthermore, all pattern matching algorithms using $N(v,l)$ extraction can benefit from the \emph{PCSR} structure.
The \emph{Prealloc-Combine} strategy also sheds new light on join optimization.

\nop{
The contribution of this work lies not only in the high-performance subgraph isomorphism algorithm, but also in the lessons we learn for optimizing graph algorithms on GPU.
Firstly, all pattern matching algorithms(e.g., frequent mining) using $N(v,l)$ can benefit from the \emph{PCSR} structure.
Secondly, \emph{Prealloc-Combine} strategy opens new door to joining process, which is superior to traditional two-step output scheme.
Thirdly, other strategies such as write cache and load balance scheme also shed new light on the optimization of highly irregular graph algorithms on GPU.
In future, we will further improve the efficiency of GSI, support extremely large graphs which can not be fully placed on GPU and develop more practical GPU graph algorithms.
}
